\documentclass[aps,pra,twocolumn,superscriptaddress,10pt,longbibliography,amsfonts,amssymb,amsmath,floatfix,floats,a4paper]{revtex4-1}
\usepackage{color}
\usepackage{amsmath}
\usepackage{amsthm}
\usepackage{amssymb}
\usepackage{amsfonts}
\usepackage{amstext}
\usepackage{graphicx}	
\usepackage{bbm}
\usepackage{enumerate}
\usepackage{tikz}
\usepackage{xr}	

\newcommand{\bra}[1]{\left\langle #1 \right\rvert}
\newcommand{\ket}[1]{ \left\lvert #1\right\rangle}
\newcommand{\braket}[2]{\left\langle #1 \middle\vert #2 \right\rangle}
\newcommand{\ketbra}[2]{\left\lvert #1 \rangle \! \langle #2 \right\rvert}

\newcommand{\kb}[2]{\left\lvert #1 \rangle \! \langle #2 \right\rvert}

\newcommand{\set}[1]{\left\lbrace #1 \right\rbrace}

\newcommand{\norm}[1]{\left\lVert #1 \right\rVert}
\newcommand{\abs}[1]{\left\lvert #1 \right\rvert}

\newcommand{\1}{\mathbbm{1}}
\newcommand{\tr}[1]{\mathrm{tr}\left( #1 \right)}
\newcommand{\partr}[2]{\mathrm{tr}_{#1}\left( #2 \right)}	

\def\submission{}	
\ifdefined\submission
  \newcommand{\graphicsPath}{./}
\else
  \newcommand{\graphicsPath}{./graphics/}
\fi 

\def\showMain{}	
\def\showSM{}	

\ifdefined\showMain
\else
\fi

\theoremstyle{plain}
\newtheorem{thm}{Theorem}
\newtheorem{lem}[thm]{Lemma}

\theoremstyle{definition}

\theoremstyle{remark}

\newcommand{\bbC}{\ensuremath{\mathbb{C}}}

\newcommand{\bbF}{\ensuremath{\mathbb{F}}}

\newcommand{\bbR}{\ensuremath{\mathbb{R}}}

\definecolor{orange}{rgb}{1,0.5,0}
\definecolor{darkgreen}{rgb}{0.0, 0.578125, 0.25}

\begin{document}

\title{Measurement-Disturbance Tradeoff Outperforming Optimal Cloning}
\centerline{}
\centerline{}
\author{Lukas Knips}
\affiliation{Max-Planck-Institut f\"{u}r Quantenoptik, Hans-Kopfermann-Stra{\ss}e 1, 85748 Garching, Germany}
\affiliation{Department f\"{u}r Physik, Ludwig-Maximilians-Universit\"{a}t, 80797 M\"{u}nchen, Germany}
\author{Jan Dziewior}
\affiliation{Max-Planck-Institut f\"{u}r Quantenoptik, Hans-Kopfermann-Stra{\ss}e 1, 85748 Garching, Germany}
\affiliation{Department f\"{u}r Physik, Ludwig-Maximilians-Universit\"{a}t, 80797 M\"{u}nchen, Germany}
\author{Anna-Lena K. Hashagen}
\affiliation{Fakult\"{a}t f\"{u}r Mathematik, Technische Universit\"{a}t M\"{u}nchen, Germany}
\author{Jasmin D. A. Meinecke}
\affiliation{Max-Planck-Institut f\"{u}r Quantenoptik, Hans-Kopfermann-Stra{\ss}e 1, 85748 Garching, Germany}
\affiliation{Department f\"{u}r Physik, Ludwig-Maximilians-Universit\"{a}t, 80797 M\"{u}nchen, Germany}
\author{Harald Weinfurter}
\affiliation{Max-Planck-Institut f\"{u}r Quantenoptik, Hans-Kopfermann-Stra{\ss}e 1, 85748 Garching, Germany}
\affiliation{Department f\"{u}r Physik, Ludwig-Maximilians-Universit\"{a}t, 80797 M\"{u}nchen, Germany}
\author{Michael M. Wolf}
\affiliation{Fakult\"{a}t f\"{u}r Mathematik, Technische Universit\"{a}t M\"{u}nchen, Germany}

\ifdefined\showMain

\begin{abstract}
One of the characteristic features of quantum mechanics is that every measurement that extracts information about a general quantum system necessarily causes an unavoidable disturbance to the state of this system. 
%
A plethora of different approaches has been developed to characterize and optimize this tradeoff. 
%
Here, we apply the framework of quantum instruments to investigate the optimal tradeoff and to derive a class of procedures that is optimal with respect to most meaningful measures.
We focus our analysis on binary measurements on qubits as commonly used in communication and computation protocols and demonstrate theoretically and in an experiment that the optimal universal asymmetric quantum cloner, albeit ideal for cloning, is not an optimal procedure for measurements and can be outperformed with high significance.  
\end{abstract}

\maketitle

\textit{Introduction.---}The work of Heisenberg, best visualized by the Heisenberg microscope~\cite{Heisenberg1930}, teaches us that every measurement is accompanied by a fundamental disturbance of a quantum system. 
The question about the precise relation between the information gained about the quantum system and the resulting disturbance has since inspired numerous studies~\cite{Jaeger1995,Englert1996,Ozawa2003,Branciard2013,Busch2013,Banaszek2001,Fuchs1996,Maccone2006,Fuchs1996,Maccone2006,Buscemi2008,Buscemi2014,Zhang2016c,DAriano2003,DAriano2003,Jordan2010,Cheong2012,NielsenGold,Kretschmann2008,Fan2015,Shitara2016}. 
A central problem is to find a tight, quantitative tradeoff relation, e.g., for the maximally achievable information for a given disturbance or, vice versa, for the minimal disturbance for a certain amount of extracted information. 
Obviously, this is not only relevant for quantum foundations, but also for many applications in quantum communication~\cite{Gisin2002,Pan2012} and quantum computation~\cite{Ekert1996,Vedral1998,Steane1998}. 
Initially studied in the context of which-path information and loss of visibility in interferometers~\cite{Jaeger1995,Englert1996}, quantifying the information-disturbance tradeoff was based on various measures such as the traditional root mean squared distance~\cite{Ozawa2003,Branciard2013}, the distance of probability distributions \cite{Busch2013}, operation and estimation fidelities~\cite{Banaszek2001,Fuchs1996,Maccone2006}, entropic quantities~\cite{Fuchs1996,Maccone2006,Buscemi2008,Buscemi2014,Zhang2016c,DAriano2003}, reversibility~\cite{DAriano2003,Jordan2010,Cheong2012}, stabilized operator norms~\cite{NielsenGold,Kretschmann2008}, state discrimination probability~\cite{Buscemi2008}, probability distribution fidelity~\cite{Fan2015}, and Fisher information~\cite{Shitara2016}. 
In spite of all these distinct approaches, no clear candidate for a most fundamental framework for the analysis of the information-disturbance tradeoff in quantum mechanics has yet emerged. 

Here we build upon a novel, comprehensive information-disturbance relation introduced recently by two of us~\cite{Hashagen_Wolf_2018}. 
There, optimal measurement devices have been proven to be independent of the chosen quality measures, as long as these fulfill some reasonable assumptions, such as convexity and basis-independence. 
This approach is unique with respect to the employment of reference observables.
On one hand, since information eventually is obtained via measurements of observables, we base the quantification of the measurement error on a reference observable.
On the other hand, the measurement induced disturbance is defined without relying on any reference observable in order not to restrict the further usage of the post-measurement state.
For a finite-dimensional von Neumann measurement, the optimal tradeoff can be achieved with quantum instruments described by at most two parameters.

\begin{figure}[ht]
\includegraphics[width=0.45\textwidth]{\graphicsPath 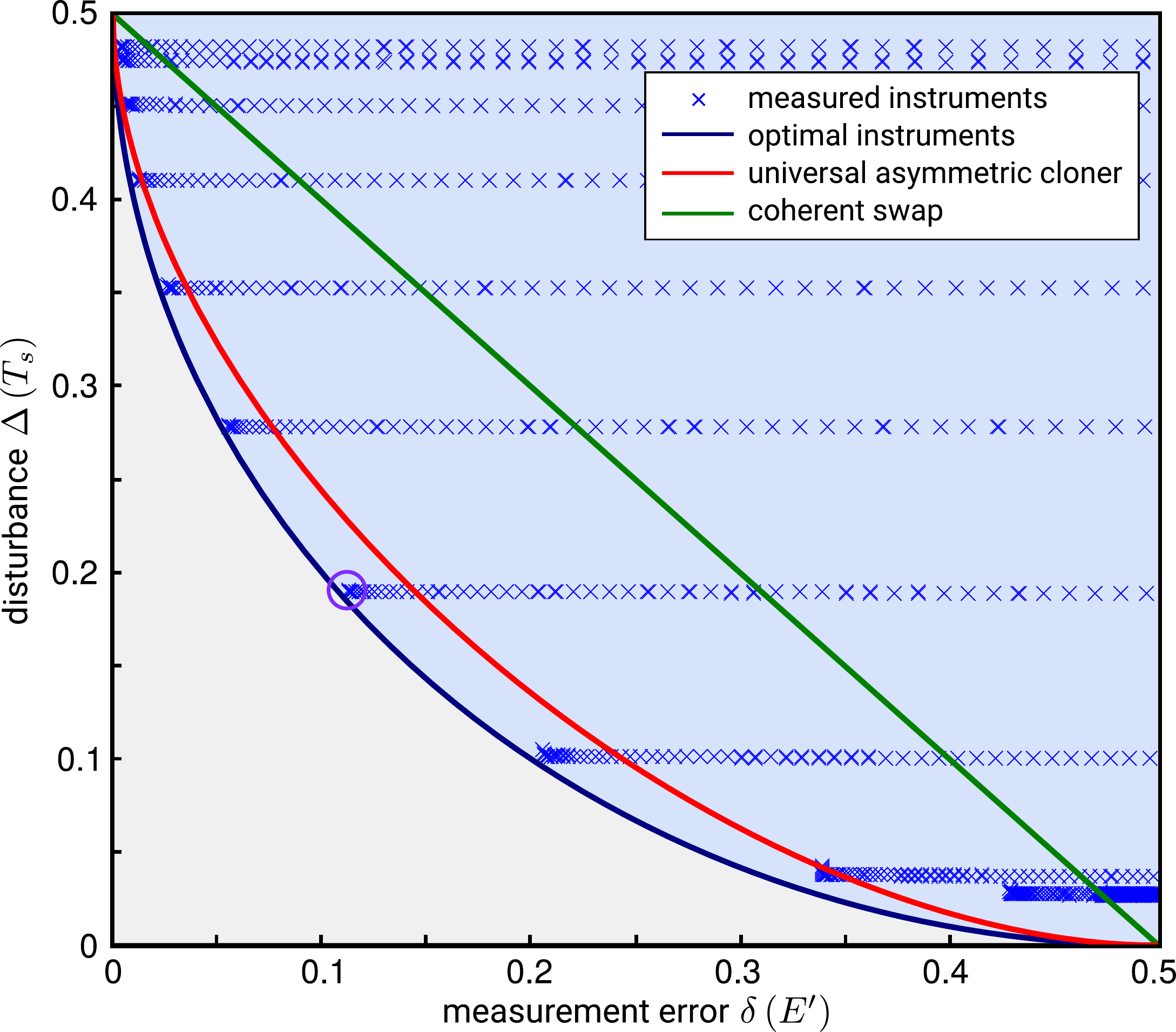}
\caption{{The optimal quantum instruments in terms of measurement error and disturbance} clearly outperform the optimal asymmetric cloner (red curve) and the coherent swap operation (green line).
Our measurements (blue crosses) come close to the theoretical curve (blue curve).
The violet marked instrument is discussed in Fig.~\ref{fig:dataplot} in more detail.
The error bars are too small to be visible; for a detailed discussion see~\cite{SM}.
}
\label{fig:dataplotInstruments}
\end{figure}

In this letter, we describe how optimal instruments can be derived for typical measures of measurement error, i.e., inverse information, and state disturbance and how they can be implemented in an experiment. 
Typically, quantum cloning is considered to be a good choice to achieve an optimal measurement disturbance tradeoff. 
Yet, here we show that the optimal instruments outperform all (asymmetric) quantum cloners~\cite{SM}. 
We test the tradeoff relation experimentally using a tunable Mach-Zehnder-Interferometer and implement a large range of quantum instruments.
We apply these instruments to a two-dimensional quantum system encoded in the photon polarization and investigate the relation between the error of the measurement and the disturbance of the qubit state. 
As distance measures we consider exemplarily some of the measures recommended in \cite{NielsenGold}, i.e., the worst-case total variational distance and the worst-case trace norm. 
For other measures see supplemental material (SM)~\cite{SM}. 
The experiment clearly shows that the optimal universal asymmetric cloner as well as the coherent swap scheme are suboptimal (Fig.~\ref{fig:dataplotInstruments}).

\textit{Measurements as quantum instruments.---}To generally quantify both the measurement error and the measurement induced disturbance, we describe the measurement of observables on a quantum system by means of quantum instruments~\cite{Davies1970,Watrous2018} as illustrated in Fig.~\ref{fig:setup}.
Formally, a quantum instrument $I$ is defined as a set of completely positive linear maps $I:=\{ I_j \}_{j=1}^m$ that fulfills the normalization condition $\sum_{j=1}^m I_j^\ast(\mathbbm{1}) =\mathbbm{1}$, where $I_j^\ast$ denotes the dual map to $I_j$ with respect to the Hilbert-Schmidt inner product.
This description naturally encompasses the connection between the observable given by a positive operator valued measure (POVM) $E^\prime:=\{E_j^\prime\}_{j=1}^m$ and the quantum channel $T_s$, which describes the measurement induced change of the state.

In general, a quantum channel is a completely positive trace preserving linear map. 
In the context of quantum instruments, the channel is given by the sum of the linear maps with $T_s:=\sum_{j=1}^m I_j$, where each map corresponds to one measurement operator $E^\prime_j$ of the POVM.
The normalization condition of the quantum instrument ensures that the corresponding quantum channel is trace-preserving.
Expressing the channel in terms of $I$ as above reflects the decohering effect of the measurement on the quantum state of the measured system.

The measurement operators $\{E^\prime_j \}_{j=1}^m$ themselves are fully determined by $I$ via $E^\prime_j:=I_j^\ast(\mathbbm{1})$, where the probability distribution for outcomes $\{ j \}_{j=1}^{m}$ on state $\rho$ is given by $ \tr{I_j(\rho)} = \tr{I_j(\rho)\mathbbm{1}} = \tr{\rho I_j^\ast(\mathbbm{1})} = \tr{\rho E_j^\prime}$.
From this point of view, the normalization condition of the quantum instrument ensures that the distribution $\{\tr{E^\prime_j \rho} \}_{j=1}^m$ is normalized.
The instrument description based on the normalized set of maps $I$, which implies the pair $(E^\prime, T_s)$, is sufficient to exhaustively describe all possible quantum measurement processes.

\begin{figure}[ht]
	\centering
	  \begin{tikzpicture}[thick, every text node part/.style={align=center}, scale=0.8]
		\draw (0,0) rectangle (2,2);
		\node (in) [left] at (-2,1) {$\rho$};
		\node (I) at (1,1) {$I$};
		\node (out1) [right] at (4,1.5) {$T_s(\rho)$};
		\node (out2) [right] at (4,0.5) {$\{\tr{E^\prime_j \rho} \}_{j=1}^m$};
		\node (into) [left] at (0,1) {};
		\node (outof1) [right] at (2,1.5) {};
		\node (outof2) [right] at (2,0.5) {};
		\draw [->] (in) -- (into);
		\draw [->] (outof1) -- (out1);
		\draw [->, dashed] (outof2) -- (out2);
		\end{tikzpicture}
\caption{
{General description of a measurement using a quantum instrument $I$.}
Obtaining information about the quantum state via the POVM $E^\prime$ (dashed line, classical output) induces a change of the quantum state described by the quantum channel $T_s$ (solid line, quantum output).}
\label{fig:setup}
\end{figure}
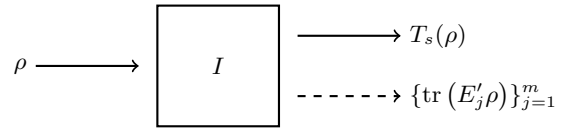


\textit{Distance measures.---}From the notion of quantum instruments it becomes immediately clear that $E^\prime$ and $T_s$  are not independent, i.e. the change of the state has a fundamental dependence on the information gained and vice versa. 
To enable a thorough quantitative analysis of this measurement-disturbance tradeoff, we use distance measures to assess the quality of the approximate measurement and to quantify the disturbance.
We quantify the disturbance $\Delta$ caused to the system by the deviation of the channel $T_s$ from the identity channel $T_{\rm id}\left(\rho\right):=\rho$. 
The measurement error $\delta$ quantifies the deviation of the measurement $E^\prime$ from a reference measurement $E$.
This approach utilizes a reference POVM $E$ to quantify the measurement error, but not the disturbance, in contrast to all other approaches found in the literature, where either a reference system is used for both, measurement error and disturbance, or none is used at all.

The measurement error $\delta$ can be quantified by defining a worst-case total variational distance based on the $l_1$-distance between probability distributions.
The $l_1$-distance, also called total variational distance, displays the largest possible difference between the probabilities that two probability distributions assign to the same event and therefore is the relevant distance measure for hypothesis testing~\cite{Neyman1933,Watrous2018}.
In our case, these two probability distributions stem from the target measurement $E$ and the actual measurement $E^\prime$ for some quantum state.
To generalize the measure for the measurement error to take into account all possible quantum states $\rho$ of the system we additionally take the worst case w.r.t. all states, which is natural when considering the maximal difference, i.e., worst-case characteristic of the $l_1$-distance itself.
Thus our worst-case total variational distance is defined as 
\begin{equation}
\delta(E^\prime) := \sup_{\rho} \frac{1}{2} \sum_{i=1}^2 \abs{\tr{E^\prime_i \rho} - \tr{E_i \rho}}. 
\label{eq:MeasError1}
\end{equation}

The quantum analogue of the worst-case total variational distance is the worst-case trace norm distance, which we thus use to quantify the distance between the quantum channel $T_s$ and the identity channel $T_{\rm id}$, 
\begin{equation}
\Delta(T_s) := \frac{1}{2} \sup_{\rho} \norm{T_s(\rho) -\rho}_1.
\label{eq:Dist}
\end{equation}
This disturbance measure quantifies how well the quantum channel $T_s$ can be distinguished from the identity channel $T_{\rm id}$ in a statistical experiment, if no auxiliary systems are allowed
\footnote{
Allowing auxiliary systems, the relevant disturbance measure is the diamond norm,
$\Delta_\diamond(T_s) := \frac{1}{2} \sup_{\xi} \norm{\left(\left(T_s - T_{{\rm id},d} \right) \otimes T_{{\rm id},d} \right)(\xi) }_1$,
where the state $\xi$ includes auxiliary systems.
Here, for the optimal tradeoff curve, the trace norm turns out to be equal to the diamond norm distance \cite{Hashagen_Wolf_2018}.
}.

\textit{Optimal instruments and tradeoff.---}As reference measurement, we choose the ideal projective measurement of the qubit with $E = \set{\ketbra{j}{j}}^2_{j=1}$. 
As proven in \cite{Hashagen_Wolf_2018} for the optimal quantum instruments each element $I_j$ can be expressed by a single Kraus operator, agreeing with the intuition that additional Kraus operators introduce noise to the system.
In the case of a qubit this leads to 
\begin{equation}
T_s(\rho) = \sum_{j=1}^2 K_j \rho K_j^\dagger \quad \text{ and } \quad \{E^\prime_j = K_j^\dagger K_j\}_{j=1}^2.
\label{eq:optimalKrausdecomposition}
\end{equation}

The Kraus operators of an optimal instrument can be chosen diagonal in the basis $\set{\ket{j}}_{j=1}^2$ given by the target measurement~\cite{Hashagen_Wolf_2018}. 
Since for a qubit there are only two of them and they must satisfy the normalization condition, in general their form is
\begin{subequations}\label{eq:generaldiagonalKraus}
\begin{align}
 K_1 = \sqrt{1-b^2_2} \ketbra{1}{1}  +  e^{i\beta_1} b_1 \ketbra{2}{2}, \\
 K_2 = b_2 \ketbra{1}{1} + e^{i\beta_2}\sqrt{1-b^2_1} \ketbra{2}{2} , 
\end{align}
\end{subequations}
with $0 \leq b_1^2,b_2^2 \leq 1$ and two arbitrary phases $\beta_1$ and $\beta_2$. 

As proven in \cite{SM}, for such an instrument, the worst-case total variational distance $\delta$ and its trace-norm analogue $\Delta$, Eqs.~(\ref{eq:MeasError1},\ref{eq:Dist}), quantifying measurement error and disturbance respectively, satisfy
\begin{equation}
\Delta \geq 
\begin{cases}
\frac{1}{2}\left( \sqrt{1-\delta} - \sqrt{\delta} \right)^2  & \text{if } \delta \leq \frac{1}{2}, \\
0 & \text{if } \delta \geq \frac{1}{2}.
\end{cases}
\label{eq:Tradeoff}
\end{equation}
The inequality is tight and cannot be exceeded by any quantum measurement procedure. 
Equality in Eq.~(\ref{eq:Tradeoff}) is attained for the family of optimal instruments defined by 
\begin{subequations}\label{eq:optimalKraus}
\begin{align}
K_1 = \frac{1}{\sqrt{2}} \left(\sqrt{1-\gamma} \ketbra{1}{1}  +  \sqrt{1+\gamma} \ketbra{2}{2}\right),  \\
K_2 =  \frac{1}{\sqrt{2}} \left(\sqrt{1+\gamma} \ketbra{1}{1}  +  \sqrt{1-\gamma} \ketbra{2}{2}\right), 
\end{align}
\end{subequations}
with $\gamma \in [0,1]$, leading to $\delta(\gamma)=\left(1-\gamma\right)/2$.

\textit{Other known measurement schemes.---}Let us evaluate common quantum measurement procedures in terms of their measurement-disturbance tradeoff. 
For perfect quantum cloning, there would be no measurement-disturbance tradeoff, as one of the perfect clones could be measured without error with the other clone staying undisturbed.
Although perfect cloning is impossible~\cite{Wootters1982}, one can derive a protocol that is optimal for approximate quantum cloning.
Hence, it is a manifest intuition that the optimal universal asymmetric quantum cloner provides a promising measurement protocol that naturally leads simultaneously to a small disturbance and a small measurement error.
It is illustrated in Fig.~\ref{fig:cloning}.
The quantum channel $T_s(\rho) = \partr{s^\prime}{T_{\text{clo}}(\rho)}$, a marginal of the cloning channel $T_{\text{clo}}$, corresponds to the evolution of the system state, obtained when tracing out the second (primed) clone.
The corresponding channel of the second clone, $T_{s^\prime}(\rho) = \partr{s}{T_{\text{clo}}(\rho)}$, provides an approximate copy to which the reference POVM $E$ is applied.
Asymmetry within the quality of the clones determines the tradeoff between the measurement error and the disturbance. 

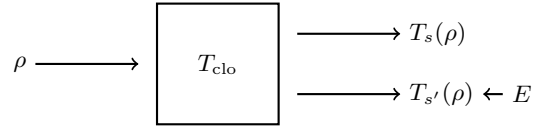
\begin{figure}[ht]
	\centering
	  \begin{tikzpicture}[thick, every text node part/.style={align=center}, scale=0.8]
		\draw (0,0) rectangle (2,2);
 		\node (in) [left] at (-2,1) {$\rho$};
		\node (T) at (1,1) {$T_{\text{clo}}$};
		\node (out1) [right] at (4,1.5) {$T_s(\rho)$};
		\node (out2) [right] at (4,0.5) {$T_{s^\prime}(\rho)$};
		\node (meas) at (6,0.5) {$E$};
 		\node (into) [left] at (0,1) {};
		\node (outof1) [right] at (2,1.5) {};
		\node (outof2) [right] at (2,0.5) {};
 		\draw [->] (in) -- (into);
		\draw [->] (outof1) -- (out1);
		\draw [->] (outof2) -- (out2);
		\draw [<-] (out2) -- (meas);
		\end{tikzpicture}
\caption{{Universal asymmetric quantum cloning.} The initial quantum state
$\rho$ is asymmetrically, approximately cloned to the auxiliary system, initially in state $\1/2$.
The target measurement is performed on one of the clones, while the other is compared to the initial quantum state $\rho$.}
\label{fig:cloning}
\end{figure}

The optimal universal asymmetric quantum cloning channel $T_\text{clo}$ for any initial quantum state $\rho$ reads~\cite{Hashagen_2017}
\begin{equation}
T_{\text{clo}}\left(\rho\right)= \left(a_2 \1 +a_1 \bbF \right)\left( \rho \otimes \frac{\1}{2}\right) \left(a_2 \1 +a_1 \bbF\right),
\label{eq:CloningChannel}
\end{equation}
with $a^2_1+a^2_2+a_1a_2 = 1$, $a_1, a_2 \in \bbR$, and the flip (or swap) operator $\bbF:= \sum_{i,j=1}^2 \ketbra{ji}{ij}$.
The parameter $a_1$ determines the amplitude of a swap operation between both qubits. 

With our measures, the measurement-disturbance tradeoff for the asymmetric quantum cloning channel satisfies
\begin{equation}
\Delta =
\begin{cases}
\frac{1}{4}\left( \sqrt{2-3\delta} - \sqrt{\delta} \right)^2  & \text{if } \delta \leq \frac{1}{2}, \\
0 & \text{if } \delta \geq \frac{1}{2}
\end{cases}
\label{eq:CloningTradeoff}
\end{equation}
with $\delta(a_2)=a_2^2/2$~\cite{SM}.

As the cloning operation cannot be realized by a unitary two-qubit transformation, any real implementation of the protocol is embedded in a larger system.
Let us thus consider an obvious analogue to the cloning operation, which can be realized by a unitary two-qubit operation.
For the swapping channel $T_{\text{cs}}$, the system interacts with the auxiliary system via a Heisenberg Hamiltonian as
\begin{align}
T_{\text{cs}}\left( \rho \right) &= e^{it\bbF} \left( \rho \otimes \tilde{\rho} \right) e^{-it\bbF} \nonumber \\
&= \left( a_2 \1 + i a_1 \bbF \right) \left( \rho \otimes \tilde{\rho} \right) \left( a_2 \1 - i a_1 \bbF \right),
\label{eq:SwapChannel}
\end{align}
with $t \in [0,\pi/2]$ or using a parametrization analogous to the cloning scheme with $a^2_1 + a^2_2 = 1$, $a_1, a_2 \in \bbR$.
The extreme cases are no swap ($t=0$, $a_2 = 1$) and full swap ($t=\pi/2$, $a_1 = 1$).

The $\delta$-$\Delta$-tradeoff for the target measurement $E=\{ \ketbra{j}{j} \}_{j=1}^2$ performed on one of the outputs satisfies
\begin{equation}
\Delta = \frac{1}{2} - \delta,
\label{eq:CohSwapTradeoff}
\end{equation}
with $\delta(t)=(1-a_1^2)/2$, 
for the coherent swap~\cite{SM}, evidently also inferior to our optimal instruments, Eq.~(\ref{eq:optimalKraus}), with the tradeoff given in Eq.~(\ref{eq:Tradeoff}).

\begin{figure}[t]
\includegraphics[width=0.47\textwidth]{\graphicsPath 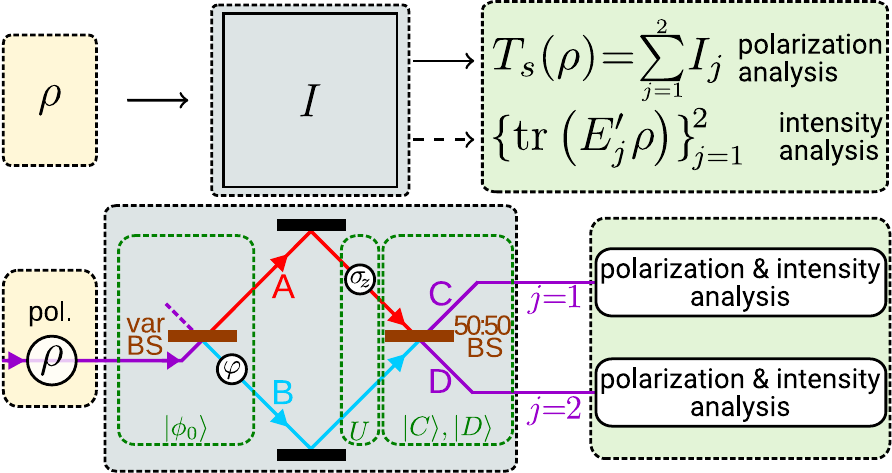}
\caption{{Conceptual experimental setup.} 
The state $\rho$ is encoded in the polarization degree of freedom of a photon, which is sent to a variable beam splitter (var BS).
The spatial superposition state inside of the interferometer is denoted by $\ket{\phi_0}$ and can be tuned in terms of relative intensities and phase. 
For the interaction $U$ between the path and the polarization degrees of freedom we apply a $\sigma_z$ operation to the polarization in one path.
Projections onto the output ports $\ket{C}$ and $\ket{D}$ of a balanced $50\!\!:\!\!50$ beam splitter conclude the realization of the Kraus operators as given in Eqs.~(\ref{eq:KrausOperatorsExp}).
Polarization and intensity measurements are performed at the output ports of the interferometer.
Please note that the actual experiment, while equivalent to the shown setup, is structured differently such that the polarization state $\rho$ is created inside of the interferometer.
The actual experiment is described in more detail in~\cite{SM}.
}
\label{fig:expsetup}
\end{figure}

\textit{Experimental implementation.---}For our experimental evaluation of the measurement-disturbance tradeoff we want to realize a broad range of quantum instruments including the optimal ones.
For that purpose we consider the polarization degree of freedom of photons to encode $\rho$, with $\ket{1}\leftrightarrow\ket{H}$ and $\ket{2}\leftrightarrow\ket{V}$, where $\ket{H}$ ($\ket{V}$) denotes horizontally (vertically) polarized light.
The Kraus operators describing the chosen set of instruments are thus given by
\begin{align} \label{eq:KrausOperators}
K_{1,2} = \frac{1}{\sqrt{2}} \Big[ \sqrt{ 1 \pm \gamma} \ketbra{H}{H} 
+ e^{i\beta} \sqrt{ 1 \mp \gamma} \ketbra{V}{V} \Big]
\end{align}
with 
an arbitrary phase $\beta$.
The optimal cases Eqs.~(\ref{eq:optimalKraus}) are achieved for $\beta = 0$.

To experimentally realize a quantum instrument and to enable analysis of the two outputs $T_s$ and $E^\prime$, it is necessary to employ an additional auxiliary quantum system, which is not yet explicitly present in the instrument description of Fig.~\ref{fig:setup}.
For the measurement of photon polarization a natural candidate is the path degree of freedom of the photons.
Since in our case a two dimensional auxiliary system is sufficient, we employ a Mach-Zehnder interferometer, which provides the two path states $\ket{A}$ and $\ket{B}$, see Fig.~\ref{fig:expsetup}.
The properties of the instrument are then determined by the initial state of this auxiliary system, $\ket{\phi_0}=\cos\alpha\ket{A}+e^{i\varphi}\sin\alpha\ket{B}$, the measurement performed on it, i.e., the detection in the output path states $\ket{C}$ and $\ket{D}$, as well as by an intermediate interaction between path and polarization.
The interaction is given by a unitary evolution $U$, which exchanges information between the systems.
We use $U= i \sigma_z \otimes \kb{A}{A} + \1 \otimes \kb{B}{B}$, which introduces a polarization dependent phase shift in arm $\ket{A}$.

For an initial path state $\ket{\phi_0}$ the Kraus operators, which act on the polarization, can then be obtained as
\begin{subequations}\label{eq:KrausOperatorsExp}
\begin{align}
K_1 &= \mathrm{tr}_\text{path} \left[ (\1 \otimes \kb{C}{C}) \, U \, (\1 \otimes \kb{\phi_0}{\phi_0}) \right], \\
K_2 &= \mathrm{tr}_\text{path} \left[ (\1 \otimes \kb{D}{D}) \, U \, (\1 \otimes \kb{\phi_0}{\phi_0}) \right].
\end{align}
\end{subequations}
Relating these expressions with Eq.~(\ref{eq:KrausOperators}), the parameters $\gamma$ and $\beta$ are given by the experimental parameters $\alpha$ and $\varphi$ by $\gamma = \sin \left(2\alpha\right) \sin \varphi$ and $\beta = \arctan\left[ \tan\left( 2 \alpha \right) \cos \varphi\right]$.
The outcome of the measurement $E^\prime$ is then obtained by determining the total intensity in the output $C$ ($E^\prime_1$) and  $D$ ($E^\prime_2$), respectively, the action of the quantum channel $T_s$ by state tomography of the polarization degree of freedom.

\begin{figure}[t]
\includegraphics[width=0.45\textwidth]{\graphicsPath 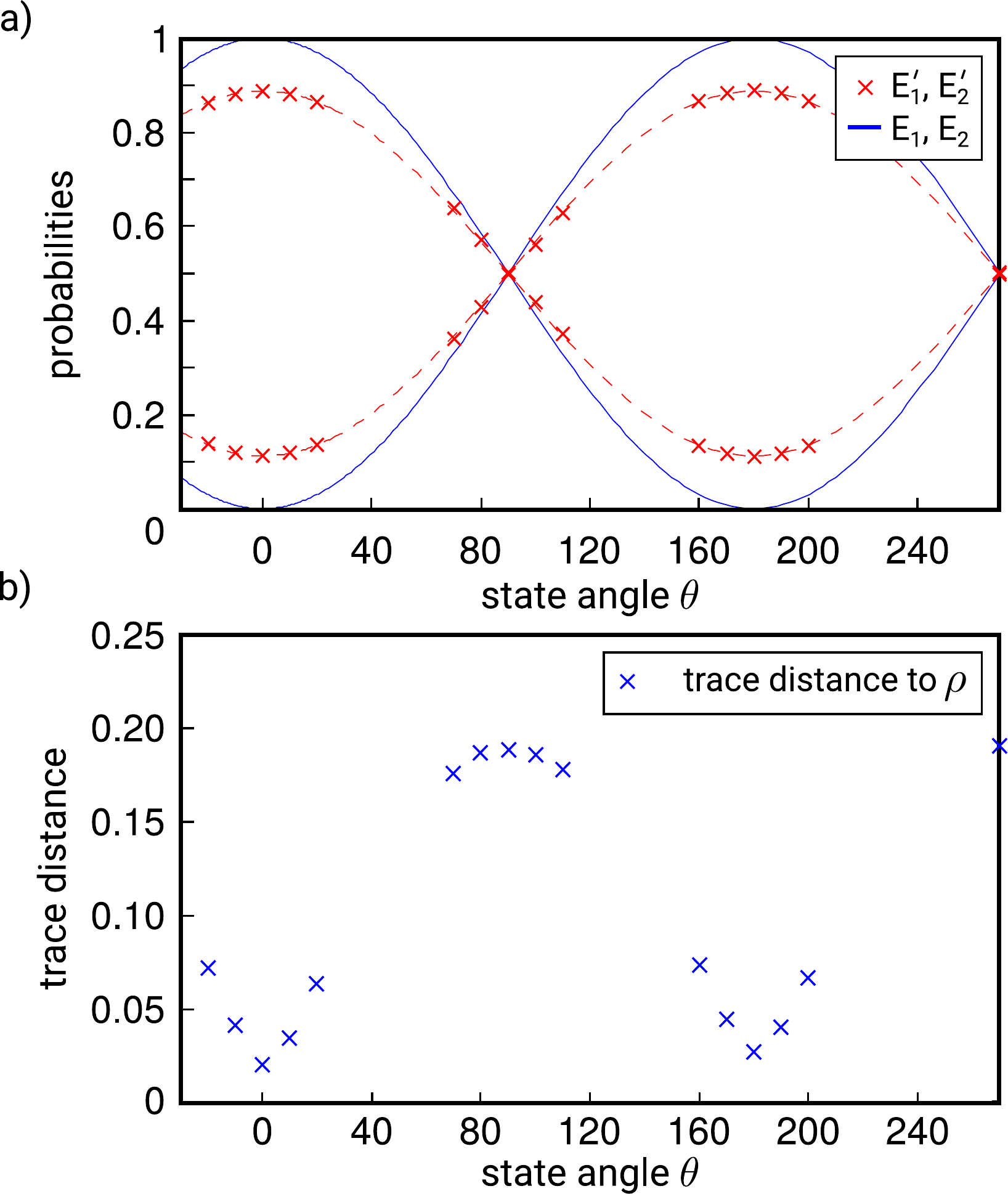}
\caption{{Evaluating measurement error $\delta$ and disturbance $\Delta$.} 
a) The measurement error corresponds to the maximal distance between the outcomes of the actual measurements $E^\prime_1$ and $E^\prime_2$ (red crosses) to the outcomes of the ideal measurements $E_1$ and $E_2$ (blue line). 
b) The disturbance is obtained by taking the supremum of the trace distance between the prepared polarization states and the tomographically reconstructed states of $T_s$. 
Please note that the suprema in a) and b) are achieved for different states. 
Statistical error bars are negligibly small.
For a detailed discussion, see ~\cite{SM}.
}
\label{fig:dataplot}
\end{figure}

\textit{Measurements and results.---}According to Eqs.~(\ref{eq:MeasError1}) and (\ref{eq:Dist}), the measures $\delta$ and $\Delta$ use the supremum over different input states $\rho$. 
We thus prepare for each quantum instrument different linearly polarized states $\rho$, which are analyzed after the interaction. 
The prepared polarization state $\rho=\ketbra{\psi}{\psi}$ in both arms is given by
$\ket{\psi} = \cos\frac{\theta}{2} \ket{H} + \sin\frac{\theta}{2} \ket{V}$,
where $\ket{H}$ and $\ket{V}$ as the eigenstates of the Pauli matrix $\sigma_z$ with eigenvalues $+1$ and $-1$, respectively, denote horizontal and vertical polarization.
We use $16$ different values for $\theta$, including those where extremal behavior for the disturbance or the measurement error is expected. 
The set of pure, linearly polarized states is sufficient as the suprema in Eqs.~(\ref{eq:MeasError1}) and (\ref{eq:Dist}) are attained in our experimental implementation, see SM~\cite{SM}. 

An intuitive strategy consists of setting a specific instrument and then varying the polarization state $\rho$, which however requires to keep the instrument parameters ($\alpha$ and $\varphi$) stable.
It turns out to be experimentally more favorable to prepare different polarization states $\rho$ and then vary the phase $\varphi$ for fixed $\alpha$ and $\rho$. 
One thus associates measurements which correspond to the same state $\ket{\phi_0}$ of the auxiliary system to the same instrument. 

The evaluation of the measurement error and the disturbance for one instrument of Fig.~\ref{fig:dataplotInstruments} is shown in Fig.~\ref{fig:dataplot} a) and b), respectively.
The supremum over a great circle of the Bloch sphere, described by $\ket{\psi}$, has been used for the analysis. 
The measurement error is given by the maximal deviation of the measurement (red crosses) to the best fitting target measurement (blue solid line), see Eq.~(\ref{eq:MeasError1}).
While some states as eigenstates of the transformation (theoretically) do not show any disturbance, for the disturbance, the largest trace distance has to be taken into account, see Eq.~(\ref{eq:Dist}).

The obtained values for measurement error and state disturbance are shown in Fig.~\ref{fig:dataplotInstruments} for the set of experimentally prepared quantum instruments. 
Each data point here identifies one quantum instrument, for which the supremum of the prepared quantum states in terms of measurement error and disturbance is determined. 
The horizontal structure is explained when considering that for a fixed $\alpha$, various measurements with different $\varphi$ have been taken, see Eq.~\eqref{eq:KrausOperators}.
We could show that there exist quantum instruments, also experimentally accessible, which significantly outperform the optimal universal asymmetric cloner (red curve) and the coherent swap operation (green line) in terms of the considered distances.

\textit{Conclusion.---}We applied the novel approach derived in \cite{Hashagen_Wolf_2018} to the setting of binary qubit measurements achieving an optimal measurement-disturbance tradeoff.
In this setting a reference measurement is used to quantitatively obtain the measurement error. 
The disturbance, on the other hand, does not depend on any reference measurement, but solely on comparing the state before and after the measurement.
Our protocol is tailored for applications based on a specific measurement without restricting subsequent use of the post-measurement state.


Furthermore, we have demonstrated that the strategies of optimal universal asymmetric quantum cloning and coherent swap do not perform optimally when considering the tradeoff relation between measurement error and disturbance.
Those protocols are optimal for their respective purposes such as approximate quantum cloning, but cannot compete with the optimal quantum instruments in the measurement scenario as in general they result in worse measurement-disturbance tradeoff relations. 
We have shown that the advantage of optimal instruments over other schemes is experimentally accessible and not only a mere theoretical improvement.
In future applications our findings allow to identify these procedures which retrieve information at the physically lowest cost in terms of state disturbance.

\textit{Acknowledgments.---}We thank Jonas Goeser for stimulating discussions.
This research was supported in part by the National Science Foundation under Grant No. NSF PHY11-25915 and by the German excellence initiative Nanosystems Initiative Munich.
LK and AKH are supported by the PhD program \textit{Exploring Quantum Matter} of the Elite Network of Bavaria.
JD acknowledges support by the International Max-Planck Research Program for Quantum Science and Technology (IMPRS-QST).
JDMA is supported by an LMU research fellowship.


\begin{thebibliography}{32}%
\makeatletter
\providecommand \@ifxundefined [1]{%
 \@ifx{#1\undefined}
}%
\providecommand \@ifnum [1]{%
 \ifnum #1\expandafter \@firstoftwo
 \else \expandafter \@secondoftwo
 \fi
}%
\providecommand \@ifx [1]{%
 \ifx #1\expandafter \@firstoftwo
 \else \expandafter \@secondoftwo
 \fi
}%
\providecommand \natexlab [1]{#1}%
\providecommand \enquote  [1]{``#1''}%
\providecommand \bibnamefont  [1]{#1}%
\providecommand \bibfnamefont [1]{#1}%
\providecommand \citenamefont [1]{#1}%
\providecommand \href@noop [0]{\@secondoftwo}%
\providecommand \href [0]{\begingroup \@sanitize@url \@href}%
\providecommand \@href[1]{\@@startlink{#1}\@@href}%
\providecommand \@@href[1]{\endgroup#1\@@endlink}%
\providecommand \@sanitize@url [0]{\catcode `\\12\catcode `\$12\catcode
  `\&12\catcode `\#12\catcode `\^12\catcode `\_12\catcode `\%12\relax}%
\providecommand \@@startlink[1]{}%
\providecommand \@@endlink[0]{}%
\providecommand \url  [0]{\begingroup\@sanitize@url \@url }%
\providecommand \@url [1]{\endgroup\@href {#1}{\urlprefix }}%
\providecommand \urlprefix  [0]{URL }%
\providecommand \Eprint [0]{\href }%
\providecommand \doibase [0]{http://dx.doi.org/}%
\providecommand \selectlanguage [0]{\@gobble}%
\providecommand \bibinfo  [0]{\@secondoftwo}%
\providecommand \bibfield  [0]{\@secondoftwo}%
\providecommand \translation [1]{[#1]}%
\providecommand \BibitemOpen [0]{}%
\providecommand \bibitemStop [0]{}%
\providecommand \bibitemNoStop [0]{.\EOS\space}%
\providecommand \EOS [0]{\spacefactor3000\relax}%
\providecommand \BibitemShut  [1]{\csname bibitem#1\endcsname}%
\let\auto@bib@innerbib\@empty
\bibitem [{\citenamefont {Heisenberg}(1930)}]{Heisenberg1930}%
  \BibitemOpen
  \bibfield  {author} {\bibinfo {author} {\bibfnamefont {Werner}\ \bibnamefont
  {Heisenberg}},\ }\href@noop {} {\emph {\bibinfo {title} {{The Physical
  Principles of the Quantum Theory}}}}\ (\bibinfo  {publisher} {University of
  Chicago Press},\ \bibinfo {year} {1930})\BibitemShut {NoStop}%
\bibitem [{\citenamefont {Jaeger}\ \emph {et~al.}(1995)\citenamefont {Jaeger},
  \citenamefont {Shimony},\ and\ \citenamefont {Vaidman}}]{Jaeger1995}%
  \BibitemOpen
  \bibfield  {author} {\bibinfo {author} {\bibfnamefont {Gregg}\ \bibnamefont
  {Jaeger}}, \bibinfo {author} {\bibfnamefont {Abner}\ \bibnamefont {Shimony}},
  \ and\ \bibinfo {author} {\bibfnamefont {Lev}\ \bibnamefont {Vaidman}},\
  }\bibfield  {title} {\enquote {\bibinfo {title} {{Two interferometric
  complementarities}},}\ }\href {\doibase 10.1103/PhysRevA.51.54} {\bibfield
  {journal} {\bibinfo  {journal} {Phys. Rev. A}\ }\textbf {\bibinfo {volume}
  {51}},\ \bibinfo {pages} {54--67} (\bibinfo {year} {1995})}\BibitemShut
  {NoStop}%
\bibitem [{\citenamefont {Englert}(1996)}]{Englert1996}%
  \BibitemOpen
  \bibfield  {author} {\bibinfo {author} {\bibfnamefont {Berthold-Georg}\
  \bibnamefont {Englert}},\ }\bibfield  {title} {\enquote {\bibinfo {title}
  {{Fringe Visibility and Which-Way Information: An Inequality}},}\ }\href
  {\doibase 10.1103/PhysRevLett.77.2154} {\bibfield  {journal} {\bibinfo
  {journal} {Phys. Rev. Lett.}\ }\textbf {\bibinfo {volume} {77}},\ \bibinfo
  {pages} {2154--2157} (\bibinfo {year} {1996})}\BibitemShut {NoStop}%
\bibitem [{\citenamefont {Ozawa}(2003)}]{Ozawa2003}%
  \BibitemOpen
  \bibfield  {author} {\bibinfo {author} {\bibfnamefont {Masanao}\ \bibnamefont
  {Ozawa}},\ }\bibfield  {title} {\enquote {\bibinfo {title} {{Universally
  valid reformulation of the Heisenberg uncertainty principle on noise and
  disturbance in measurement}},}\ }\href {\doibase 10.1103/PhysRevA.67.042105}
  {\bibfield  {journal} {\bibinfo  {journal} {Phys. Rev. A}\ }\textbf {\bibinfo
  {volume} {67}},\ \bibinfo {pages} {042105} (\bibinfo {year}
  {2003})}\BibitemShut {NoStop}%
\bibitem [{\citenamefont {Branciard}(2013)}]{Branciard2013}%
  \BibitemOpen
  \bibfield  {author} {\bibinfo {author} {\bibfnamefont {Cyril}\ \bibnamefont
  {Branciard}},\ }\bibfield  {title} {\enquote {\bibinfo {title}
  {{Error-tradeoff and error-disturbance relations for incompatible quantum
  measurements}},}\ }\href {\doibase 10.1073/pnas.1219331110} {\bibfield
  {journal} {\bibinfo  {journal} {Proceedings of the National Academy of
  Sciences}\ }\textbf {\bibinfo {volume} {110}},\ \bibinfo {pages} {6742--6747}
  (\bibinfo {year} {2013})}\BibitemShut {NoStop}%
\bibitem [{\citenamefont {Busch}\ \emph {et~al.}(2013)\citenamefont {Busch},
  \citenamefont {Lahti},\ and\ \citenamefont {Werner}}]{Busch2013}%
  \BibitemOpen
  \bibfield  {author} {\bibinfo {author} {\bibfnamefont {Paul}\ \bibnamefont
  {Busch}}, \bibinfo {author} {\bibfnamefont {Pekka}\ \bibnamefont {Lahti}}, \
  and\ \bibinfo {author} {\bibfnamefont {Reinhard~F.}\ \bibnamefont {Werner}},\
  }\bibfield  {title} {\enquote {\bibinfo {title} {{Proof of Heisenberg's
  Error-Disturbance Relation}},}\ }\href {\doibase
  10.1103/PhysRevLett.111.160405} {\bibfield  {journal} {\bibinfo  {journal}
  {Phys. Rev. Lett.}\ }\textbf {\bibinfo {volume} {111}},\ \bibinfo {pages}
  {160405} (\bibinfo {year} {2013})}\BibitemShut {NoStop}%
\bibitem [{\citenamefont {Banaszek}(2001)}]{Banaszek2001}%
  \BibitemOpen
  \bibfield  {author} {\bibinfo {author} {\bibfnamefont {Konrad}\ \bibnamefont
  {Banaszek}},\ }\bibfield  {title} {\enquote {\bibinfo {title} {{Fidelity
  Balance in Quantum Operations}},}\ }\href {\doibase
  10.1103/PhysRevLett.86.1366} {\bibfield  {journal} {\bibinfo  {journal}
  {Phys. Rev. Lett.}\ }\textbf {\bibinfo {volume} {86}},\ \bibinfo {pages}
  {1366--1369} (\bibinfo {year} {2001})}\BibitemShut {NoStop}%
\bibitem [{\citenamefont {Fuchs}\ and\ \citenamefont
  {Peres}(1996)}]{Fuchs1996}%
  \BibitemOpen
  \bibfield  {author} {\bibinfo {author} {\bibfnamefont {Christopher~A.}\
  \bibnamefont {Fuchs}}\ and\ \bibinfo {author} {\bibfnamefont {Asher}\
  \bibnamefont {Peres}},\ }\bibfield  {title} {\enquote {\bibinfo {title}
  {{Quantum-state disturbance versus information gain: Uncertainty relations
  for quantum information}},}\ }\href {\doibase 10.1103/PhysRevA.53.2038}
  {\bibfield  {journal} {\bibinfo  {journal} {Phys. Rev. A}\ }\textbf {\bibinfo
  {volume} {53}},\ \bibinfo {pages} {2038--2045} (\bibinfo {year}
  {1996})}\BibitemShut {NoStop}%
\bibitem [{\citenamefont {Maccone}(2006)}]{Maccone2006}%
  \BibitemOpen
  \bibfield  {author} {\bibinfo {author} {\bibfnamefont {Lorenzo}\ \bibnamefont
  {Maccone}},\ }\bibfield  {title} {\enquote {\bibinfo {title}
  {{Information-disturbance tradeoff in quantum measurements}},}\ }\href
  {\doibase 10.1103/PhysRevA.73.042307} {\bibfield  {journal} {\bibinfo
  {journal} {Phys. Rev. A}\ }\textbf {\bibinfo {volume} {73}},\ \bibinfo
  {pages} {042307} (\bibinfo {year} {2006})}\BibitemShut {NoStop}%
\bibitem [{\citenamefont {Buscemi}\ \emph {et~al.}(2008)\citenamefont
  {Buscemi}, \citenamefont {Hayashi},\ and\ \citenamefont
  {Horodecki}}]{Buscemi2008}%
  \BibitemOpen
  \bibfield  {author} {\bibinfo {author} {\bibfnamefont {Francesco}\
  \bibnamefont {Buscemi}}, \bibinfo {author} {\bibfnamefont {Masahito}\
  \bibnamefont {Hayashi}}, \ and\ \bibinfo {author} {\bibfnamefont {Micha\l{}}\
  \bibnamefont {Horodecki}},\ }\bibfield  {title} {\enquote {\bibinfo {title}
  {{Global Information Balance in Quantum Measurements}},}\ }\href {\doibase
  10.1103/PhysRevLett.100.210504} {\bibfield  {journal} {\bibinfo  {journal}
  {Phys. Rev. Lett.}\ }\textbf {\bibinfo {volume} {100}},\ \bibinfo {pages}
  {210504} (\bibinfo {year} {2008})}\BibitemShut {NoStop}%
\bibitem [{\citenamefont {Buscemi}\ \emph {et~al.}(2014)\citenamefont
  {Buscemi}, \citenamefont {Hall}, \citenamefont {Ozawa},\ and\ \citenamefont
  {Wilde}}]{Buscemi2014}%
  \BibitemOpen
  \bibfield  {author} {\bibinfo {author} {\bibfnamefont {Francesco}\
  \bibnamefont {Buscemi}}, \bibinfo {author} {\bibfnamefont {Michael J.~W.}\
  \bibnamefont {Hall}}, \bibinfo {author} {\bibfnamefont {Masanao}\
  \bibnamefont {Ozawa}}, \ and\ \bibinfo {author} {\bibfnamefont {Mark~M.}\
  \bibnamefont {Wilde}},\ }\bibfield  {title} {\enquote {\bibinfo {title}
  {{Noise and Disturbance in Quantum Measurements: An Information-Theoretic
  Approach}},}\ }\href {\doibase 10.1103/PhysRevLett.112.050401} {\bibfield
  {journal} {\bibinfo  {journal} {Phys. Rev. Lett.}\ }\textbf {\bibinfo
  {volume} {112}},\ \bibinfo {pages} {050401} (\bibinfo {year}
  {2014})}\BibitemShut {NoStop}%
\bibitem [{\citenamefont {Zhang}\ \emph {et~al.}(2016)\citenamefont {Zhang},
  \citenamefont {Zhang},\ and\ \citenamefont {Yu}}]{Zhang2016c}%
  \BibitemOpen
  \bibfield  {author} {\bibinfo {author} {\bibfnamefont {Jun}\ \bibnamefont
  {Zhang}}, \bibinfo {author} {\bibfnamefont {Yang}\ \bibnamefont {Zhang}}, \
  and\ \bibinfo {author} {\bibfnamefont {Chang-Shui}\ \bibnamefont {Yu}},\
  }\bibfield  {title} {\enquote {\bibinfo {title} {{The Measurement-Disturbance
  Relation and the Disturbance Trade-off Relation in Terms of Relative
  Entropy}},}\ }\bibfield  {booktitle} {\emph {\bibinfo {booktitle}
  {International Journal of Theoretical Physics}},\ }\href@noop {} {\ \textbf
  {\bibinfo {volume} {55}} (\bibinfo {year} {2016})}\BibitemShut {NoStop}%
\bibitem [{\citenamefont {D'Ariano}(2003)}]{DAriano2003}%
  \BibitemOpen
  \bibfield  {author} {\bibinfo {author} {\bibfnamefont {Giacomo~M.}\
  \bibnamefont {D'Ariano}},\ }\bibfield  {title} {\enquote {\bibinfo {title}
  {{On the Heisenberg principle, namely on the information-disturbance
  trade-off in a quantum measurement}},}\ }\href {\doibase
  10.1002/prop.200310045} {\bibfield  {journal} {\bibinfo  {journal}
  {Fortschritte der Physik}\ }\textbf {\bibinfo {volume} {51}},\ \bibinfo
  {pages} {318--330} (\bibinfo {year} {2003})}\BibitemShut {NoStop}%
\bibitem [{\citenamefont {Jordan}\ and\ \citenamefont
  {Korotkov}(2010)}]{Jordan2010}%
  \BibitemOpen
  \bibfield  {author} {\bibinfo {author} {\bibfnamefont {Andrew~N.}\
  \bibnamefont {Jordan}}\ and\ \bibinfo {author} {\bibfnamefont {Alexander~N.}\
  \bibnamefont {Korotkov}},\ }\bibfield  {title} {\enquote {\bibinfo {title}
  {{Uncollapsing the wavefunction by undoing quantum measurements}},}\ }\href
  {\doibase 10.1080/00107510903385292} {\bibfield  {journal} {\bibinfo
  {journal} {Contemporary Physics}\ }\textbf {\bibinfo {volume} {51}},\
  \bibinfo {pages} {125--147} (\bibinfo {year} {2010})}\BibitemShut {NoStop}%
\bibitem [{\citenamefont {Cheong}\ and\ \citenamefont
  {Lee}(2012)}]{Cheong2012}%
  \BibitemOpen
  \bibfield  {author} {\bibinfo {author} {\bibfnamefont {Yong~Wook}\
  \bibnamefont {Cheong}}\ and\ \bibinfo {author} {\bibfnamefont {Seung-Woo}\
  \bibnamefont {Lee}},\ }\bibfield  {title} {\enquote {\bibinfo {title}
  {{Balance Between Information Gain and Reversibility in Weak Measurement}},}\
  }\href {\doibase 10.1103/PhysRevLett.109.150402} {\bibfield  {journal}
  {\bibinfo  {journal} {Phys. Rev. Lett.}\ }\textbf {\bibinfo {volume} {109}},\
  \bibinfo {pages} {150402} (\bibinfo {year} {2012})}\BibitemShut {NoStop}%
\bibitem [{\citenamefont {Gilchrist}\ \emph {et~al.}(2005)\citenamefont
  {Gilchrist}, \citenamefont {Langford},\ and\ \citenamefont
  {Nielsen}}]{NielsenGold}%
  \BibitemOpen
  \bibfield  {author} {\bibinfo {author} {\bibfnamefont {Alexei}\ \bibnamefont
  {Gilchrist}}, \bibinfo {author} {\bibfnamefont {Nathan~K.}\ \bibnamefont
  {Langford}}, \ and\ \bibinfo {author} {\bibfnamefont {Michael~A.}\
  \bibnamefont {Nielsen}},\ }\bibfield  {title} {\enquote {\bibinfo {title}
  {{Distance measures to compare real and ideal quantum processes}},}\ }\href
  {\doibase 10.1103/PhysRevA.71.062310} {\bibfield  {journal} {\bibinfo
  {journal} {Phys. Rev. A}\ }\textbf {\bibinfo {volume} {71}},\ \bibinfo
  {pages} {062310} (\bibinfo {year} {2005})}\BibitemShut {NoStop}%
\bibitem [{\citenamefont {Kretschmann}\ \emph {et~al.}(2008)\citenamefont
  {Kretschmann}, \citenamefont {Schlingemann},\ and\ \citenamefont
  {Werner}}]{Kretschmann2008}%
  \BibitemOpen
  \bibfield  {author} {\bibinfo {author} {\bibfnamefont {Dennis}\ \bibnamefont
  {Kretschmann}}, \bibinfo {author} {\bibfnamefont {Dirk}\ \bibnamefont
  {Schlingemann}}, \ and\ \bibinfo {author} {\bibfnamefont {Reinhard~F.}\
  \bibnamefont {Werner}},\ }\bibfield  {title} {\enquote {\bibinfo {title}
  {{The Information-Disturbance Tradeoff and the Continuity of Stinespring's
  Representation}},}\ }\href {\doibase 10.1109/TIT.2008.917696} {\bibfield
  {journal} {\bibinfo  {journal} {IEEE Transactions on Information Theory}\
  }\textbf {\bibinfo {volume} {54}},\ \bibinfo {pages} {1708--1717} (\bibinfo
  {year} {2008})}\BibitemShut {NoStop}%
\bibitem [{\citenamefont {Fan}\ \emph {et~al.}(2015)\citenamefont {Fan},
  \citenamefont {Ge}, \citenamefont {Nha},\ and\ \citenamefont
  {Zubairy}}]{Fan2015}%
  \BibitemOpen
  \bibfield  {author} {\bibinfo {author} {\bibfnamefont {Longfei}\ \bibnamefont
  {Fan}}, \bibinfo {author} {\bibfnamefont {Wenchao}\ \bibnamefont {Ge}},
  \bibinfo {author} {\bibfnamefont {Hyunchul}\ \bibnamefont {Nha}}, \ and\
  \bibinfo {author} {\bibfnamefont {M.~S.}\ \bibnamefont {Zubairy}},\
  }\bibfield  {title} {\enquote {\bibinfo {title} {{Trade-off between
  information gain and fidelity under weak measurements}},}\ }\href {\doibase
  10.1103/PhysRevA.92.022114} {\bibfield  {journal} {\bibinfo  {journal} {Phys.
  Rev. A}\ }\textbf {\bibinfo {volume} {92}},\ \bibinfo {pages} {022114}
  (\bibinfo {year} {2015})}\BibitemShut {NoStop}%
\bibitem [{\citenamefont {Shitara}\ \emph {et~al.}(2016)\citenamefont
  {Shitara}, \citenamefont {Kuramochi},\ and\ \citenamefont
  {Ueda}}]{Shitara2016}%
  \BibitemOpen
  \bibfield  {author} {\bibinfo {author} {\bibfnamefont {Tomohiro}\
  \bibnamefont {Shitara}}, \bibinfo {author} {\bibfnamefont {Yui}\ \bibnamefont
  {Kuramochi}}, \ and\ \bibinfo {author} {\bibfnamefont {Masahito}\
  \bibnamefont {Ueda}},\ }\bibfield  {title} {\enquote {\bibinfo {title}
  {{Trade-off relation between information and disturbance in quantum
  measurement}},}\ }\href {\doibase 10.1103/PhysRevA.93.032134} {\bibfield
  {journal} {\bibinfo  {journal} {Phys. Rev. A}\ }\textbf {\bibinfo {volume}
  {93}},\ \bibinfo {pages} {032134} (\bibinfo {year} {2016})}\BibitemShut
  {NoStop}%
\bibitem [{\citenamefont {Gisin}\ \emph {et~al.}(2002)\citenamefont {Gisin},
  \citenamefont {Ribordy}, \citenamefont {Tittel},\ and\ \citenamefont
  {Zbinden}}]{Gisin2002}%
  \BibitemOpen
  \bibfield  {author} {\bibinfo {author} {\bibfnamefont {Nicolas}\ \bibnamefont
  {Gisin}}, \bibinfo {author} {\bibfnamefont {Gr{\'{e}}goire}\ \bibnamefont
  {Ribordy}}, \bibinfo {author} {\bibfnamefont {Wolfgang}\ \bibnamefont
  {Tittel}}, \ and\ \bibinfo {author} {\bibfnamefont {Hugo}\ \bibnamefont
  {Zbinden}},\ }\bibfield  {title} {\enquote {\bibinfo {title} {{Quantum
  cryptography}},}\ }\href {\doibase 10.1103/revmodphys.74.145} {\bibfield
  {journal} {\bibinfo  {journal} {Reviews of Modern Physics}\ }\textbf
  {\bibinfo {volume} {74}},\ \bibinfo {pages} {145--195} (\bibinfo {year}
  {2002})}\BibitemShut {NoStop}%
\bibitem [{\citenamefont {Pan}\ \emph {et~al.}(2012)\citenamefont {Pan},
  \citenamefont {Chen}, \citenamefont {Lu}, \citenamefont {Weinfurter},
  \citenamefont {Zeilinger},\ and\ \citenamefont {{\.{Z}}ukowski}}]{Pan2012}%
  \BibitemOpen
  \bibfield  {author} {\bibinfo {author} {\bibfnamefont {Jian-Wei}\
  \bibnamefont {Pan}}, \bibinfo {author} {\bibfnamefont {Zeng-Bing}\
  \bibnamefont {Chen}}, \bibinfo {author} {\bibfnamefont {Chao-Yang}\
  \bibnamefont {Lu}}, \bibinfo {author} {\bibfnamefont {Harald}\ \bibnamefont
  {Weinfurter}}, \bibinfo {author} {\bibfnamefont {Anton}\ \bibnamefont
  {Zeilinger}}, \ and\ \bibinfo {author} {\bibfnamefont {Marek}\ \bibnamefont
  {{\.{Z}}ukowski}},\ }\bibfield  {title} {\enquote {\bibinfo {title}
  {{Multiphoton entanglement and interferometry}},}\ }\href {\doibase
  10.1103/revmodphys.84.777} {\bibfield  {journal} {\bibinfo  {journal}
  {Reviews of Modern Physics}\ }\textbf {\bibinfo {volume} {84}},\ \bibinfo
  {pages} {777--838} (\bibinfo {year} {2012})}\BibitemShut {NoStop}%
\bibitem [{\citenamefont {Ekert}\ and\ \citenamefont
  {Jozsa}(1996)}]{Ekert1996}%
  \BibitemOpen
  \bibfield  {author} {\bibinfo {author} {\bibfnamefont {Artur}\ \bibnamefont
  {Ekert}}\ and\ \bibinfo {author} {\bibfnamefont {Richard}\ \bibnamefont
  {Jozsa}},\ }\bibfield  {title} {\enquote {\bibinfo {title} {{Quantum
  computation and Shor's factoring algorithm}},}\ }\href {\doibase
  10.1103/revmodphys.68.733} {\bibfield  {journal} {\bibinfo  {journal}
  {Reviews of Modern Physics}\ }\textbf {\bibinfo {volume} {68}},\ \bibinfo
  {pages} {733--753} (\bibinfo {year} {1996})}\BibitemShut {NoStop}%
\bibitem [{\citenamefont {Vedral}\ and\ \citenamefont
  {Plenio}(1998)}]{Vedral1998}%
  \BibitemOpen
  \bibfield  {author} {\bibinfo {author} {\bibfnamefont {Vlatko}\ \bibnamefont
  {Vedral}}\ and\ \bibinfo {author} {\bibfnamefont {Martin~B.}\ \bibnamefont
  {Plenio}},\ }\bibfield  {title} {\enquote {\bibinfo {title} {{Basics of
  quantum computation}},}\ }\href {\doibase 10.1016/s0079-6727(98)00004-4}
  {\bibfield  {journal} {\bibinfo  {journal} {Progress in Quantum Electronics}\
  }\textbf {\bibinfo {volume} {22}},\ \bibinfo {pages} {1--39} (\bibinfo {year}
  {1998})}\BibitemShut {NoStop}%
\bibitem [{\citenamefont {Steane}(1998)}]{Steane1998}%
  \BibitemOpen
  \bibfield  {author} {\bibinfo {author} {\bibfnamefont {Andrew}\ \bibnamefont
  {Steane}},\ }\bibfield  {title} {\enquote {\bibinfo {title} {{Quantum
  computing}},}\ }\href {\doibase 10.1088/0034-4885/61/2/002} {\bibfield
  {journal} {\bibinfo  {journal} {Reports on Progress in Physics}\ }\textbf
  {\bibinfo {volume} {61}},\ \bibinfo {pages} {117--173} (\bibinfo {year}
  {1998})}\BibitemShut {NoStop}%
\bibitem [{\citenamefont {{Hashagen}}\ and\ \citenamefont
  {{Wolf}}(2018)}]{Hashagen_Wolf_2018}%
  \BibitemOpen
  \bibfield  {author} {\bibinfo {author} {\bibfnamefont {Anna-Lena~K.}\
  \bibnamefont {{Hashagen}}}\ and\ \bibinfo {author} {\bibfnamefont
  {Michael~M.}\ \bibnamefont {{Wolf}}},\ }\bibfield  {title} {\enquote
  {\bibinfo {title} {{Universality and Optimality in the
  Information-Disturbance Tradeoff}},}\ }\href@noop {} {\bibfield  {journal}
  {\bibinfo  {journal} {ArXiv e-prints}\ } (\bibinfo {year} {2018})},\ \Eprint
  {http://arxiv.org/abs/1802.09893} {arXiv:1802.09893 [quant-ph]} \BibitemShut
  {NoStop}%
\bibitem [{SM()}]{SM}%
  \BibitemOpen
  \href@noop {} {\emph {\bibinfo {title} {{Supplemental
  Material}}}}\BibitemShut {NoStop}%
\bibitem [{\citenamefont {Davies}\ and\ \citenamefont
  {Lewis}(1970)}]{Davies1970}%
  \BibitemOpen
  \bibfield  {author} {\bibinfo {author} {\bibfnamefont {E.~Brian}\
  \bibnamefont {Davies}}\ and\ \bibinfo {author} {\bibfnamefont {John~T.}\
  \bibnamefont {Lewis}},\ }\bibfield  {title} {\enquote {\bibinfo {title} {{An
  operational approach to quantum probability}},}\ }\href {\doibase
  10.1007/BF01647093} {\bibfield  {journal} {\bibinfo  {journal}
  {Communications in Mathematical Physics}\ }\textbf {\bibinfo {volume} {17}},\
  \bibinfo {pages} {239--260} (\bibinfo {year} {1970})}\BibitemShut {NoStop}%
\bibitem [{\citenamefont {Watrous}(2018)}]{Watrous2018}%
  \BibitemOpen
  \bibfield  {author} {\bibinfo {author} {\bibfnamefont {John}\ \bibnamefont
  {Watrous}},\ }\href {https://cs.uwaterloo.ca/~watrous/TQI/TQI.pdf} {\emph
  {\bibinfo {title} {{The Theory of Quantum Information}}}}\ (\bibinfo
  {publisher} {Cambridge University Press},\ \bibinfo {year}
  {2018})\BibitemShut {NoStop}%
\bibitem [{\citenamefont {Neyman}\ and\ \citenamefont
  {Pearson}(1933)}]{Neyman1933}%
  \BibitemOpen
  \bibfield  {author} {\bibinfo {author} {\bibfnamefont {Jerzy}\ \bibnamefont
  {Neyman}}\ and\ \bibinfo {author} {\bibfnamefont {Egon~S.}\ \bibnamefont
  {Pearson}},\ }\bibfield  {title} {\enquote {\bibinfo {title} {{On the Problem
  of the Most Efficient Tests of Statistical Hypotheses}},}\ }\href {\doibase
  10.1098/rsta.1933.0009} {\bibfield  {journal} {\bibinfo  {journal}
  {Philosophical Transactions of the Royal Society A: Mathematical, Physical
  and Engineering Sciences}\ }\textbf {\bibinfo {volume} {231}},\ \bibinfo
  {pages} {289--337} (\bibinfo {year} {1933})}\BibitemShut {NoStop}%
\bibitem [{Note1()}]{Note1}%
  \BibitemOpen
  \bibinfo {note} {Allowing auxiliary systems, the relevant disturbance measure
  is the diamond norm, $\Delta _\diamond (T_s) := \protect \frac {1}{2}
  \protect \qopname \relax m{sup}_{\xi } \left \delimiter 69645069 \left (\left
  (T_s - T_{{\protect \rm id},d} \right ) \otimes T_{{\protect \rm id},d}
  \right )(\xi ) \right \delimiter 86422285 _1$, where the state $\xi $
  includes auxiliary systems. Here, for the optimal tradeoff curve, the trace
  norm turns out to be equal to the diamond norm distance \cite
  {Hashagen_Wolf_2018}.}\BibitemShut {Stop}%
\bibitem [{\citenamefont {Wootters}\ and\ \citenamefont
  {Zurek}(1982)}]{Wootters1982}%
  \BibitemOpen
  \bibfield  {author} {\bibinfo {author} {\bibfnamefont {William~K.}\
  \bibnamefont {Wootters}}\ and\ \bibinfo {author} {\bibfnamefont
  {Wojciech~H.}\ \bibnamefont {Zurek}},\ }\bibfield  {title} {\enquote
  {\bibinfo {title} {A single quantum cannot be cloned},}\ }\href {\doibase
  10.1038/299802a0} {\bibfield  {journal} {\bibinfo  {journal} {Nature}\
  }\textbf {\bibinfo {volume} {299}},\ \bibinfo {pages} {802--803} (\bibinfo
  {year} {1982})}\BibitemShut {NoStop}%
\bibitem [{\citenamefont {Hashagen}(2017)}]{Hashagen_2017}%
  \BibitemOpen
  \bibfield  {author} {\bibinfo {author} {\bibfnamefont {Anna-Lena~K.}\
  \bibnamefont {Hashagen}},\ }\bibfield  {title} {\enquote {\bibinfo {title}
  {{Universal Asymmetric Quantum Cloning Revisited}},}\ }\href@noop {}
  {\bibfield  {journal} {\bibinfo  {journal} {Quant. Inf. Comp.}\ }\textbf
  {\bibinfo {volume} {17}},\ \bibinfo {pages} {0747--0778} (\bibinfo {year}
  {2017})}\BibitemShut {NoStop}%
\end{thebibliography}

\begin{thebibliography}{32}
\bibitem[25]{Hashagen_Wolf_2018}
Anna-Lena K. Hashagen and Michael M. Wolf, ``Universality and Optimality in the Information-Disturbance Tradeoff,'' ArXiv e-prints (2018), arXiv:1802.09893 [quant-ph].

\end{thebibliography}

%

\else
\fi 


\ifdefined\showSM

\newpage
\setcounter{equation}{0}
\setcounter{figure}{0}

\makeatletter
\renewcommand{\theequation}{S\@arabic\c@equation}
\renewcommand{\thefigure}{S\@arabic\c@figure}
\makeatletter

\ifdefined\showMain
\else
\title{Supplemental Material: \\ Measurement-disturbance tradeoff outperforming optimal cloning}
\maketitle
\fi 

\ifdefined\showMain
\section*{\large Supplemental Material}
\else
\fi
\section*{SM\,1: Optimal tradeoff relation} 
\label{sec:proofOptimalTradeoff}
\begin{thm}[Total variation - trace norm tradeoff]
Consider a von Neumann target measurement given by an orthonormal basis $\left\{ \ket{i} \in \bbC^2 \right\}_{i=1}^2$, and an instrument with two corresponding outcomes. 
Then the worst-case total variational distance $\delta$ and its trace-norm analogue $\Delta$, defined as in Eqs.~(\ref{eq:MeasError1},\ref{eq:Dist}), quantifying measurement error and disturbance respectively, satisfy
\begin{equation}
\Delta \geq 
\begin{cases}
\frac{1}{2}\left( \sqrt{1-\delta} - \sqrt{\delta} \right)^2  & \text{if } \delta \leq \frac{1}{2}, \\
0 & \text{if } \delta \geq \frac{1}{2}.
\end{cases}
\label{eq:TradeoffAppendix}
\end{equation}
The inequality is tight and equality is attained within the family of instruments defined by  
\begin{equation}
I_j(\rho) := K_j\rho K_j, \qquad j=1,2,
\label{eq:OptInstr}
\end{equation}
with 
\begin{equation}
K_{1,2} = \frac{1}{\sqrt{2}} \left(\sqrt{1\pm\gamma} \ketbra{1}{1}  +  \sqrt{1\mp\gamma} \ketbra{2}{2}\right)
\end{equation}
with $\gamma \in [0,1]$. 
\label{thm:theoremOptimalTradeoff}
\end{thm}

\begin{proof}
In order to derive the information-disturbance tradeoff, we need to solve the following optimization problem: \\
For $\gamma \in [0,1]$ 
\begin{align}
\label{eq:opt1}
&\text{minimize } & & \Delta\left(T_s = \sum_{j=1}^2 I_j\right) \\
&\text{subject to } & & \delta\left(E^\prime = \left\{ I_j^\ast(\1) \right\}_{j=1}^2 \right) \leq \gamma, \nonumber \\
&&& I_j \text{ is c.p. and} \nonumber \\
&&& \sum_{j=1}^2 I_j^\ast (\1) = \1, \nonumber
\end{align}
where the last two constraints ensure that $I$ is an instrument. As discussed before, we assume that every element of the instrument can be expressed using a single Kraus operator. This agrees well with intuition, because more Kraus operators introduce more noise to the system. Furthermore, we assume that these Kraus operators can be chosen diagonal in the basis of the target measurement, $E=\{ \ketbra{j}{j} \}_{i=1}^2$, to reflect the symmetry of the optimization problem. These assumptions simplify the optimization problem significantly.
The Kraus operators given in Eq.~(\ref{eq:generaldiagonalKraus}) then yield the following POVM elements of the approximate measurement
\begin{equation}
E^\prime_j = (1 - b_{\bar{j}}^2) \ketbra{j}{j} + b^2_j(\1 -\ketbra{j}{j}),
\label{eq:EffectOp}
\end{equation}
for $j=1,2$, where $\bar{j}=2$ if $j=1$ and $\bar{j}=1$ if $j=2$ with $0 \leq b_1^2,b_2^2 \leq 1$.
The measurement error is thus given as
\begin{align*}
\delta(E^\prime) &= \sup_{\rho} \frac{1}{2}\sum_{j=1}^2 \abs{\tr{E_j' \rho} - \bra{j} \rho \ket{j}} \\
&= \sup_{\rho} \frac{1}{2}\sum_{j=1}^2 \abs{\tr{ \left( b_j^2 \1 - (b_j^2+b_{\bar{j}}^2) \ketbra{j}{j} \right) \rho} } \\
&= \sup_{\norm{\psi}=1} \frac{1}{2}\sum_{j=1}^2 \abs{ \bra{\psi} b_j^2 \1 - (b_j^2+b_{\bar{j}}^2) \ketbra{j}{j} \ket{\psi} } \\
&= \frac{1}{2}(b_1^2 + b_2^2),
\end{align*}
where the convexity of the $l_1$-norm was used. 
The disturbance follows from direct calculations,
\begin{align*}
\Delta(T_1) &= \frac{1}{2} \sup_{\rho} \norm{T_1(\rho)-\rho}_1 \\
&= \frac{1}{2} \sup_{\rho} \norm{\sum_{j=1}^2K_j\rho K_j^\dagger-\rho}_1 \\
&= \frac{1}{2}\left|1-e^{i\beta_1}b_1\sqrt{1-b_2^2} - e^{i\beta_2}b_2\sqrt{1-b_1^2} \right|.
\end{align*}
Without loss of generality, we may assume that $b_1,b_2\geq 0$ in the optimization problem, such that an optimum is attained for $\beta_1=\beta_2=0$.
The optimization problem given in Eq.~(\ref{eq:opt1}) therefore simplifies: \\
For $\gamma \in [0,1]$ 
\begin{align}
\label{eq:opt2}
&\text{minimize } & & \frac{1}{2}\left(1-b_1\sqrt{1-b_2^2} - b_2\sqrt{1-b_1^2} \right) \\
&\text{subject to } & & \frac{1}{2}(b_1^2 + b_2^2) \leq \frac{1}{2}\left(1-\gamma\right), \nonumber \\
&&& 0 \leq b_1,b_2 \leq 1. \nonumber 
\end{align}
The global minimum is achieved at 
\begin{equation*}
b_1=b_2=
\begin{cases}
\sqrt{\frac{1}{2}} & \gamma \in \left[ -1, 0\right] \\
\sqrt{\frac{1}{2}}\sqrt{1-\gamma} & \gamma \in \left[ 0, 1\right] 
\end{cases}
\end{equation*}
and as stated in Eq.~(\ref{eq:TradeoffAppendix}).
\end{proof}

\section*{SM\,2: Tradeoff relation for optimal universal asymmetric cloning}\label{sec:derivationCloning}
\begin{thm}[Total variation - trace norm tradeoff using optimal universal asymmetric cloning]
Consider a von Neumann measurement given by an orthonormal basis in $\bbC^2$ on one of the outputs of the optimal universal $1\to 2$ asymmetric quantum cloning channel. 
Then the worst-case total variational distance $\delta$ and its trace-norm analogue $\Delta$ satisfy
\begin{equation}
\Delta =
\begin{cases}
\frac{1}{4}\left( \sqrt{2-3\delta} - \sqrt{\delta} \right)^2  & \text{if } \delta \leq \frac{1}{2}, \\
0 & \text{if } \delta \geq \frac{1}{2}.
\end{cases}
\label{eq:CloningTradeoffAppendix}
\end{equation}
\label{thm:theoremCloningTradeoff}
\end{thm}

\begin{proof}
The marginals of the optimal cloning channel are given by
\begin{equation}
T_{\text{clo},i} (\rho) = a_i^2 \frac{\1}{2} \tr{\rho} + (1-a_i^2)\rho, \ \ i=1,2,
\end{equation}
with $T_{\text{clo},1}=T_s$ and $T_{\text{clo},2}=T_{s^\prime}$.
The  marginal quantum channel $T_s$ describes the evolution of the quantum state and its distance to the identity channel $T_{\rm id}$ then quantifies the disturbance.
Similarly, the marginal $T_{s^\prime}$, whose output is measured by the target measurement $E$, describes the measurement itself through $E_j' = T_{s^\prime}^\ast(E_j)$. 
This is illustrated in Fig.~\ref{fig:cloning}.
This yields for the disturbance
\begin{align*}
\Delta(T_s) :=& \frac{1}{2} \sup_{\rho} \norm{T_s(\rho) -\rho}_1 \\
=& \frac{1}{2}\sup_{\rho} \norm{a_1^2\frac{\1}{2}-a_1^2\rho}_1 \\
=& \frac{a_1^2}{2}.
\end{align*}
The measurement error turns out to be
\begin{align*}
\delta(E^\prime) :=&  \sup_{\rho} \frac{1}{2} \sum_{j=1}^2  \abs{\tr{E_j' \rho} - \bra{j} \rho \ket{j}} \\
=&  \sup_{\rho} \frac{1}{2} \sum_{j=1}^2  \abs{\tr{T_{s^\prime}^\ast(\ketbra{j}{j}) \rho} - \bra{j} \rho \ket{j}} \\
=&  \sup_{\rho} \frac{1}{2} \sum_{j=1}^2  \abs{\tr{\ketbra{j}{j} T_{s^\prime}(\rho)} - \bra{j} \rho \ket{j}} \\
=&  \sup_{\rho} \frac{1}{2} \sum_{j=1}^2  \abs{\bra{j} a_2^2 \frac{\1}{2} - a_2^2 \rho \ket{j}} \\
=& \frac{a_2^2}{2}.
\end{align*}
Substituting this into the trace-preserving condition of the optimal universal asymmetric quantum cloning channel, we obtain the theorem~\ref{thm:theoremCloningTradeoff}. 
\end{proof}

\section*{SM\,3: Tradeoff relation for coherent swap}\label{sec:derivationSwap}

\begin{thm}[Total variation - trace norm tradeoff using the coherent swap]
Consider a von Neumman measurement given by an orthonormal basis in $\bbC^2$ on one of the outputs of a coherent swap channel. 
Then the worst-case total variational distance $\delta$ and its trace-norm analogue $\Delta$ satisfy
\begin{equation}
\Delta = \frac{1}{2} - \delta.
\label{eq:CohSwapTradeoffAppendix}
\end{equation}
\label{thm:theoremSwapTradeoff}
\end{thm}

\begin{proof}
Using the substitution $a_1=a$ and $a_2 = \sqrt{1-a^2}$ with $a\in [0,1]$
yields the two marginals of the coherent swap quantum channel,
\begin{equation}
T_s(\rho) = a^2 \tilde{\rho} + (1-a^2)\rho 
\label{eq:T1swap}
\end{equation}
and
\begin{equation}
T_{s^\prime}(\rho) = (1-a^2)\tilde{\rho}+a^2 \rho.
\label{eq:T2swap}
\end{equation}
The disturbance is therefore 
\begin{align*}
\Delta(T_s) :=& \frac{1}{2} \sup_{\rho} \norm{T_s(\rho) -\rho}_1 \\
=& \frac{1}{2} a^2 \sup_{\rho} \norm{\tilde{\rho} -\rho}_1.
\end{align*}
The optimal choice for $\tilde{\rho}$ should clearly satisfy the points $(\Delta(T_s)=0,\delta(E^\prime)=1/2)$ and $(\Delta(T_s)=1/2,\delta(E^\prime)=0)$, where again $E^\prime=T_{s^\prime}^\ast(E)$.
For any such choice of $\tilde{\rho}$ the disturbance thus satisfies $\Delta(T_s) \geq a^2/2$. 
The measurement error turns out to be
\begin{align*}
\delta(E^\prime) :=&  \sup_{\rho} \frac{1}{2} \sum_{j=1}^2  \abs{\tr{E_j^\prime \rho} - \bra{j} \rho \ket{j}} \\
=& \sup_{\rho} \frac{1}{2} \sum_{j=1}^2  \abs{\tr{T_{s^\prime}^\ast(\ketbra{j}{j}) \rho} - \bra{j} \rho \ket{j}} \\
=& \sup_{\rho} \frac{1}{2} \sum_{j=1}^2  \abs{\tr{\ketbra{j}{j} T_{s^\prime}(\rho)} - \bra{j} \rho \ket{j}} \\
=&  \left(1-a^2\right)\sup_{\rho} \frac{1}{2} \sum_{j=1}^2  \abs{ \bra{j} \tilde{\rho} \ket{j} - \bra{j} \rho \ket{j} }.
\end{align*}
Thus, an optimal choice for $\tilde{\rho}$ that minimizes the disturbance and the measurement error is $\tilde{\rho} = \1/2$. 
A pure state with the same diagonal entries yields the same measurement error; it would, however, increase the disturbance caused to the system. 

The disturbance is then
\begin{equation*}
\Delta(T_s) = \frac{a^2}{2},
\end{equation*}
and the measurement error is
\begin{equation*}
\delta(E^\prime) = \frac{1}{2}\left( 1-a^2\right).
\end{equation*}
This gives the linear tradeoff curve given in theorem~\ref{thm:theoremSwapTradeoff}.
\end{proof}

\section*{SM\,4: Properties of distance measures}
\label{sec:distanceProperties}
The distance measures used throughout this manu\-script to quantify the measurement error and the disturbance, denoted by $\delta$ and $\Delta$, satisfy Assumption~1 and Assumption~2 of \cite{Hashagen_Wolf_2018} respectively. 


\begin{lem}
\label{lem:prop1}
$\delta$ as defined in Eq.~(\ref{eq:MeasError1}) satisfies the following properties:
\begin{enumerate}[(a)]
	\item $\delta(\{\ketbra{i}{i}\}_{i=1}^2) = 0$,
	\item $\delta$ is convex, 
	\item $\delta$ is permutation invariant, i.e., for every permutation $\pi$ and any measurement $M$
	\begin{equation}
	\delta \left( \{ U_\pi^\dagger M_{\pi(i)} U_\pi \}_{i=1}^2 \right) = \delta \left(\{ M_i  \}_{i=1}^2\right), \nonumber
	\end{equation}
	where $U_\pi$ is the permutation matrix that acts as $U_\pi \ket{i} = \ket{\pi(i)}$, and
	\item $\delta$ is invariant under diagonal unitaries, i.e., that for every diagonal unitary $D$ and any measurement $M$
		\begin{equation}
	\delta \left( \{ D^\dagger M_{i} D \}_{i=1}^2 \right) = \delta \left(\{ M_i  \}_{i=1}^2\right). \nonumber
	\end{equation}
\end{enumerate}
\end{lem}

\begin{proof}
Let $\delta(M):=\sup_{\rho}\frac{1}{2} \sum_{i=1}^2  \abs{\tr{M_i \rho} - \bra{i} \rho \ket{i}}$. 
Then
\begin{enumerate}[(a)]
	\item $\delta(\{\ketbra{i}{i}\}_{i=1}^2) = 0$, since
	\begin{equation*}
	\delta(\{\ketbra{i}{i}\}_{i=1}^2) = \sup_{\rho} \frac{1}{2} \sum_{i=1}^2 \abs{\bra{i} \rho \ket{i} - \bra{i} \rho \ket{i}} =0,
	\end{equation*}
	\item $\delta$ is convex, since for any measurements $M, M'$ and for all $\lambda \in [0,1]$,
	\begin{align*}
	&\delta \left(\lambda M + (1-\lambda)M' \right) \\
	= &\sup_{\rho } \frac{1}{2} \sum_{i=1}^2 \abs{\tr{\left(\lambda M_i + (1-\lambda)M_i'\right)  \rho} - \bra{i} \rho \ket{i}} \\
	\leq &  \lambda  \sup_{\rho } \frac{1}{2} \sum_{i=1}^2 \abs{ \tr{M_i\rho} -\bra{i} \rho \ket{i} }  \\
	& \qquad +  (1-\lambda) \sup_{\rho} \frac{1}{2} \sum_{i=1}^2 \abs{ \left( \tr{M_i'\rho} -\bra{i} \rho \ket{i}  \right)} \\
	= & \lambda \delta(M) + (1-\lambda) \delta(M'),
	\end{align*}
	\item $\delta$ is permutation invariant, since for every permutation $\pi$ and any measurement $M$
	\begin{align*}
	&\delta \left( \{ U_\pi^\dagger M_{\pi(i)} U_\pi \}_{i=1}^2 \right) \\
	= & \sup_{\rho} \frac{1}{2} \sum_{i=1}^2 \abs{\tr{U_\pi^\dagger M_{\pi(i)} U_\pi  \rho} - \bra{i} \rho \ket{i}} \\
	= &\sup_{\rho} \frac{1}{2} \sum_{i=1}^2 \abs{\tr{ M_{\pi(i)}   \rho } - \bra{\pi(i)} \rho  \ket{\pi(i)}} \\
	= & \sup_{\rho} \frac{1}{2} \sum_{i=1}^2 \abs{\tr{ M_{i}   \rho } - \bra{i} \rho  \ket{i}} \\
	= &\delta \left(\{ M_i  \}_{i=1}^2\right),
	\end{align*}
	where $U_\pi$ is the permutation matrix that acts as $U_\pi \ket{i} = \ket{\pi(i)}$, and
	\item $\delta$ is invariant under diagonal unitaries, since for every diagonal unitary $D$ and any measurement $M$
	\begin{align*}
	&\delta \left( \{ D^\dagger M_{i} D \}_{i=1}^2 \right) \\
	= & \sup_{\rho} \frac{1}{2} \sum_{i=1}^2 \abs{\tr{D^\dagger M_{i} D \rho} - \bra{i} \rho \ket{i}} \\
	= & \sup_{\rho} \frac{1}{2} \sum_{i=1}^2  \abs{\tr{M_{i} \rho} - \bra{i} D^\dagger \rho  D\ket{i}} \\
	= &\sup_{\rho} \frac{1}{2} \sum_{i=1}^2 \abs{\tr{M_{i} \rho} - \bra{i}  \rho  \ket{i}} \\	
	= &\delta \left(\{ M_i  \}_{i=1}^2\right).
	\end{align*}
\end{enumerate}
\end{proof}

\begin{lem}
\label{lem:prop2}
$\Delta$ as defined in Eq.~(\ref{eq:Dist}) satisfies the following properties:
\begin{enumerate}[(a)]
	\item $\Delta(T_{\rm id}) = 0$,
	\item $\Delta$ is convex, 
	\item $\Delta$ is basis-independent, i.e., for every unitary $U$ and every quantum channel $\Phi$
	\begin{equation}
	\Delta \left( U \Phi \left( U^\dagger \cdot U \right) U^\dagger \right) = \Delta\left( \Phi \right). \nonumber
	\end{equation}
\end{enumerate}
\end{lem}

\begin{proof}
Let $\Delta (\Phi) := \frac{1}{2} \sup_{\rho} \norm{\Phi(\rho) -\rho}_1$. Then
\begin{enumerate}[(a)]
	\item $\Delta(T_{\rm id}) = 0$, since $\Delta (T_{\rm id}) = \frac{1}{2} \sup_{\rho} \norm{\rho -\rho}_1 =0$,
	\item $\Delta$ is convex, since for any quantum channels $\Phi, \Phi'$ and for all $\lambda \in [0,1]$,
	\begin{align*}
	&\Delta\left( \lambda \Phi +(1-\lambda)\Phi' \right) \\
	=  &\frac{1}{2} \sup_{\rho} \norm{\left( \lambda \Phi +(1-\lambda)\Phi' \right)(\rho) -\rho}_1 \\
	= &\frac{1}{2} \sup_{\rho} \norm{ \lambda \left(\Phi(\rho) -\rho \right) +(1-\lambda) \left( \Phi'(\rho) -\rho \right) }_1 \\
	\leq & \lambda \frac{1}{2} \sup_{\rho} \norm{ \Phi(\rho) -\rho }_1	+ (1-\lambda) \frac{1}{2} \sup_{\rho} \norm{ \Phi'(\rho) -\rho }_1 \\
	= & \lambda \Delta(\Phi) + (1-\lambda) \Delta (\Phi'),
	\end{align*}
	where we have used properties of a norm and properties of a supremum of a convex functional over a convex set,
	\item $\Delta$ is basis-independent, i.e., for every unitary $U$ and every quantum channel $\Phi$
	\begin{align*}
	&\Delta \left( U \Phi \left( U^\dagger \rho U \right) U^\dagger \right) \\
	= &\frac{1}{2} \sup_{\rho} \norm{U \Phi \left( U^\dagger \rho U \right) U^\dagger -\rho}_1 \\
	= &\frac{1}{2} \sup_{\rho} \norm{U \Phi \left(  \rho\right) U^\dagger - U \rho U^\dagger}_1 \\
	= &\frac{1}{2} \sup_{\rho} \norm{ \Phi \left(  \rho\right)  -  \rho }_1 \\	
	= &\Delta\left( \Phi \right),
	\end{align*}
	where we have used the fact that the trace norm is unitarily invariant. 
\end{enumerate}
\end{proof}

\section*{SM\,5: Different measures}
The optimal instruments as explained in the main text and derived in Sec.~\ref{sec:proofOptimalTradeoff} result in optimal measurement-disturbance relations for all distance measures which satisfy the assumptions of \cite{Hashagen_Wolf_2018}.
For more details on the distance measure used in the main text see Sec.~\ref{sec:distanceProperties}.
\begin{figure}[th]
\includegraphics[width=0.48\textwidth]{\graphicsPath 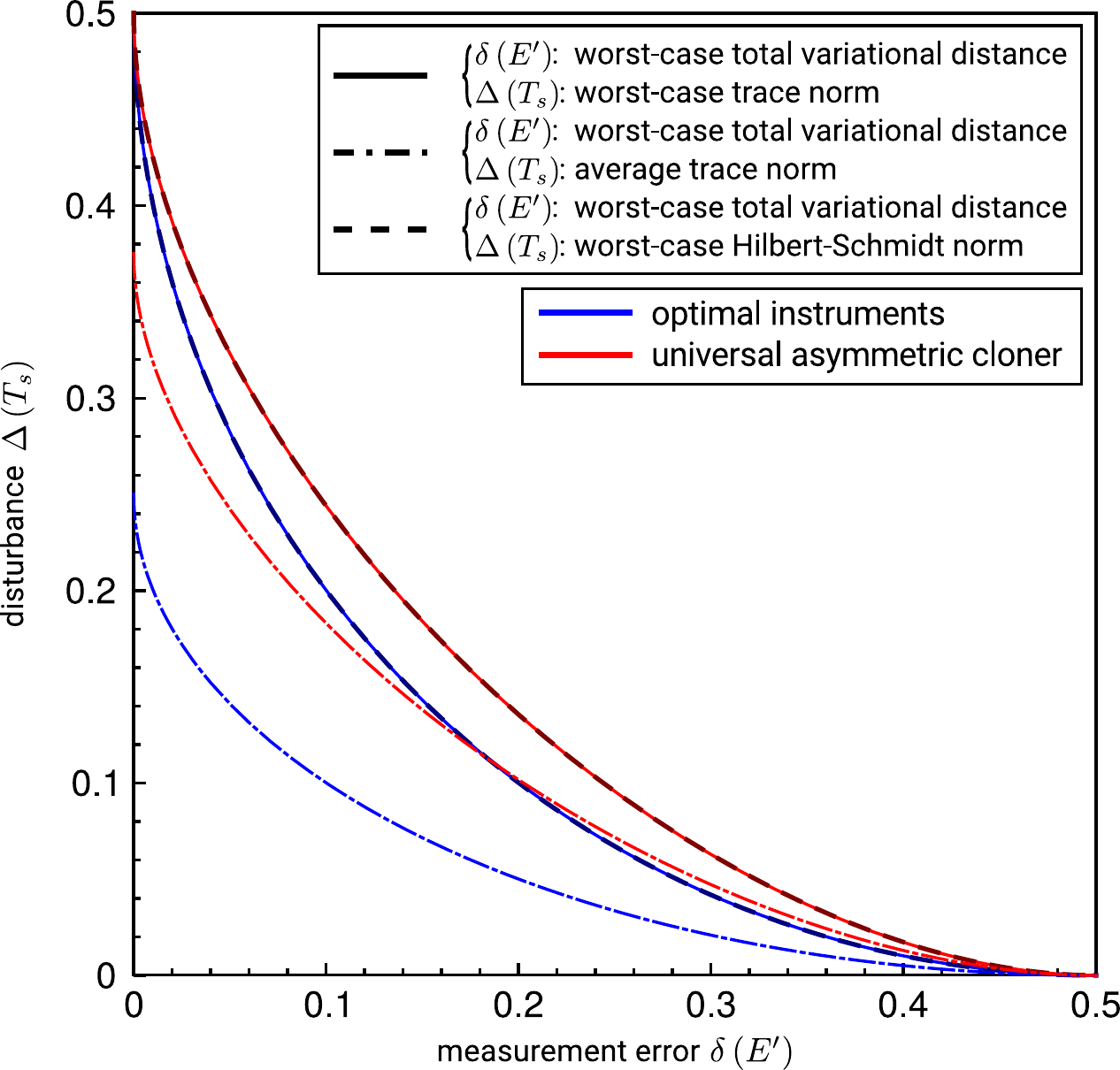}
\caption{
Comparison of optimal quantum instruments (blue) with the optimal universal asymmetric quantum cloner (red) for different distance measures based on simulations.
The tradeoff relation of the main text based on the measures of Eqs.~(\ref{eq:MeasError1}) and (\ref{eq:Dist}) is shown (solid lines) and equivalent to a properly scaled version of the worst-case Hilbert-Schmidt norm (overlayed dashed lines) and to the worst-case infidelity (not shown).
For averaging over all quantum states instead of taking the supremum of the trace norm for the disturbance, one obtains the dashdotted lines.
}
\label{fig:differentMeasures}
\end{figure}

We here show the tradeoff relations for different choices of disturbance measures, while the measurement error is always quantified as in Eq.~(\ref{eq:MeasError1}).
For various meaningful measures, we observe that the optimal instruments outperform the cloner, see Fig.~\ref{fig:differentMeasures}.

\section*{SM\,6: Experimental setup}

Due to experimental and practical limitations, the actual experimental setup has been slightly different than described in the main text.
However, the actual implementation is fully equivalent to the description there.
In order to be able to fully tune the attenuation in one of the interferometer arms, we use a half waveplate (HWP) sandwiched between two polarizers.
Therefore, the polarization state $\rho$ cannot be set before.
Hence, we decided to first create the spatial superposition state $\ket{\phi_0}$ using waveplates and polarizers and subsequently set $\rho$ in both interferometer arms separately.
With this approach, we still achieve at this stage a separable state $\rho\otimes\ketbra{\phi_0}{\phi_0}$ within the interferometer before the interaction.
As we set the polarization state directly in front of the second beam splitter of the interferometer, the reflection of beam $A$ on the beam splitter already provides the interaction between system and auxiliary system.
This reflection induces the unitary transformation $U$ as described in the main text, enabling us to obtain the Kraus operators given in Eq.~(\ref{eq:KrausOperators}).

Since for a perfect beam splitter the output ports are interchanged for $\varphi_0\leftrightarrow\varphi_0+\pi$, we use only output port $C$ to obtain data for both projections, considering the phases $\varphi_1 = \varphi_0$ and $\varphi_2 = \varphi_0 + \pi$.
This way, both projections are carried out with exactly the same equipment, reducing possible experimental errors.

\begin{figure}[t]
\includegraphics[width=0.47\textwidth]{\graphicsPath 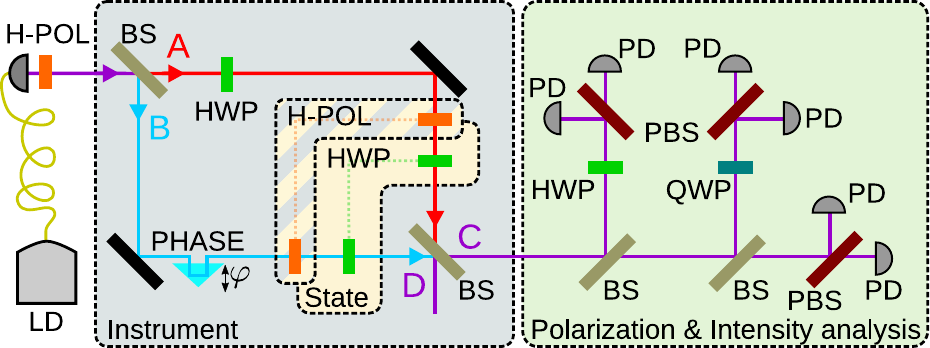}
\caption{{Actual experimental setup.} 
Light from a diode laser (LD) propagates through a single mode fiber and is sent through a fixed polarizer (H-POL).
A beam splitter (BS) creates a spatial superposition.
The attenuation of one arm can be adjusted using a half waveplate (HWP) in arm $A$ and another H-POL.
The relative phase $\varphi$ can be varied using a piezo controlled prism. 
H-POLs together with variable HWPs ensure equal polarization in both arms as indicated by the dotted lines.
As the H-POLs are used to vary the attenuation as well as to set the polarization state, they are part of both the instrument and the state preparation.
The reflection from arm $A$ on the second BS introduces a coupling between polarization and path.
Polarization and intensity measurements are performed in output port $C$ using waveplates (HWP and QWP), polarizing beam splitters (PBS) and photodiodes (PD).
Output port $D$ is not monitored, as for phase $\varphi_0$ it is redundant to the output of port $C$ at phase $\varphi_0+\pi$.
}
\label{fig:ActualExpsetup}
\end{figure}

\section*{SM\,7: Choice of polarization states}
\label{sec:choiceOfStates}
According to the parametrization 
$\ket{\psi} = \cos\frac{\theta}{2} \ket{H} + \sin\frac{\theta}{2} \ket{V}$, 
the experimentally prepared values for $\theta$ were $\{-20^\circ$, $-10^\circ$, $0^\circ$, $10^\circ$, $20^\circ$, $70^\circ$, $80^\circ$, $90^\circ$, $100^\circ$, $110^\circ$, $160^\circ$, $170^\circ$, $180^\circ$, $190^\circ$, $200^\circ$, $270^\circ\}$. 
For $\theta=0^\circ$ and $\theta=180^\circ$, the prepared state corresponds to horizontal polarization $\ket{H}$ and vertical polarization $\ket{V}$, respectively.
Thus, the reflection in beam $A$ only introduces a phase, as for example the state for $\theta=0^\circ$ is transformed according to
\begin{align}
\ket{H}&\otimes\left(\cos\alpha\ket{A}+\sin\alpha e^{i\varphi}\ket{B}\right)\rightarrow \nonumber \\
\ket{H}&\otimes\left(i\cos\alpha\ket{A}+\sin\alpha e^{i\varphi}\ket{B}\right),
\end{align}
which does not change the state of the polarization.
The disturbance therefore (ideally) vanishes.
In contrast, for $\theta=90^\circ$, we expect 
\begin{align}
\left(\ket{H}+\ket{V}\right)&\otimes\left(\cos\alpha\ket{A}+\sin\alpha e^{i\varphi}\ket{B}\right)\rightarrow \nonumber \\
i\left(\ket{H}-\ket{V}\right)&\otimes\cos\alpha\ket{A}+\left(\ket{H}+\ket{V}\right)\otimes\sin\alpha e^{i\varphi}\ket{B},
\end{align}
where normalization is omitted.
For a given instrument characterized by $\{\alpha,\varphi\}$, this polarization state is expected to give the largest disturbance $\Delta$.

For the Kraus operators given in Eq.~(\ref{eq:KrausOperators}), we find for $E_j^\prime=K_j^\dagger K_j$ for $j=1,2$,
\begin{equation}
E_{1,2}^\prime = \frac{1}{2} \begin{pmatrix}
1 \pm \sin 2\alpha\cos\varphi & 0 \\
0 & 1 \mp \sin 2\alpha\cos\varphi
\end{pmatrix}.
\end{equation}
Therefore, the distance of the outcome probabilities, used to obtain $\delta$, becomes
\begin{align}
\frac{1}{2}\sum_i\left|\tr{E_i^\prime\ketbra{\psi}{\psi}}-\left|\braket{i}{\psi}\right|^2\right|=\nonumber\\
\left|\cos\theta\left(1-\cos\varphi\sin2\alpha\right)\right|,
\end{align}
which vanishes for $\theta=90^\circ$ (and $\theta=270^\circ$) and can be maximal for $\theta=0^\circ$ (and $\theta=180^\circ$).

% ...

\section*{SM\,8: Error analysis of experimental data}
The statistical error of the data shown in Fig.~\ref{fig:dataplotInstruments} is estimated by comparing the results obtained in redundant measurements.
The standard deviation of the measurement error is estimated to be around $8.3\cdot10^{-5}$, whereas the $1\sigma$-error bar for the estimated disturbance is approximately $7.0\cdot10^{-5}$. 
Those values are thus too small to be visible in Fig.~\ref{fig:dataplotInstruments}.

Additionally to statistical errors, two different sources of systematic errors have been identified. 
First, the state preparation as well as the interaction are not perfectly implemented.  
The imperfect preparation of the initial polarization state and of the state analysis are the main reasons that the identity channel with no disturbance at all (but high measurement error) cannot be implemented perfectly, leading to a residual disturbance, which appears as an increase of the minimal disturbance $\Delta$ of the data in the plot.
In any case, this type of error only reduces the quality of the prepared quantum instruments and does not lead to faulty conclusions.

However, as a second type of systematic error one has to ensure that the prepared polarization states are describing a great circle on the Bloch sphere and contain the states with extremal results sufficiently well.
This error can be approximated by considering the data as shown in Fig.~\ref{fig:dataplot}.
By applying a parabolic model for the data points around the extrema of the probability graphs and the maxima of the trace distance graphs, the deviation of the extrema from the measured points can be estimated.
This effect might cause a quantum instrument to look better than it actually is, i.e., less disturbing together with smaller measurement error.
Yet, for the dataset shown in Fig.~\ref{fig:dataplot}~b), the parabolic fit results in a maximum at $\theta\approx89.95^\circ$ with a trace distance larger by only $0.02\%$ compared to the trace distance at  $\theta=90^\circ$.
The probabilities in Fig.~\ref{fig:dataplot}~a) around $\theta=0^\circ$ and $\theta=180^\circ$ can nicely be described by parabolae, where the extrema coincide with our measured points.
Thus, the systematic effect of underestimating the measurement error or the disturbance due to badly chosen measurement states is negligibly small.

In conclusion, the different sources of errors overall reduce the quality of the implemented quantum instruments and do not lead to an underestimation of disturbance and measurement error, respectively.
We can thus show the implementation of instruments better than the optimal quantum cloner with high significance.

\else 
\fi 

\ifdefined\showMain
\else

\onecolumngrid 

\fi 

\end{document}